\documentclass[journal]{IEEEtran}

\usepackage{cite, textcomp, amsmath, amssymb, amsfonts, amsthm, algorithm, algpseudocode}
\usepackage{graphicx, siunitx, bm, float, booktabs}
\usepackage[colorlinks=true, allcolors=black]{hyperref}
\usepackage[caption=false, font=footnotesize]{subfig}

\theoremstyle{definition}
\newtheorem{definition}{Definition}
\theoremstyle{definition}
\newtheorem{lemma}{Lemma}
\theoremstyle{definition}
\newtheorem{proposition}{Proposition}
\theoremstyle{definition}
\newtheorem{theorem}{Theorem}
\theoremstyle{definition}
\newtheorem{remark}{Remark}
\theoremstyle{definition}
\newtheorem*{problem}{Problem Statement}
\theoremstyle{definition}
\newtheorem{subproblem}{Subproblem}
\theoremstyle{definition}
\newtheorem{example}{Example}

\begin{document}

\title{Safe Navigation under Uncertain Obstacle Dynamics using Control Barrier Functions and 
Constrained Convex Generators}

\author{Hugo Matias, Daniel Silvestre
\thanks{H. Matias and D. Silvestre are with the Center of Technology and Systems (UNINOVA-CTS), 
NOVA School of Science and Technology (NOVA-FCT), 2829-516 Caparica, Portugal, and also with the 
Institute for Systems and Robotics (ISR-Lisbon), Instituto Superior Técnico (IST), 1049-001 
Lisbon, Portugal (emails: h.matias@campus.fct.unl.pt and dsilvestre@fct.unl.pt).}
\thanks{This research was partially supported by the Portuguese Fundação para a Ciência e a 
Tecnologia (FCT) through the CTS under grant UID/00066/2025, and by the European Innovation 
Council (EIC) through the ISR under project (1801P.01521.1.01) gEICko-101223148-IST-ID.}}

\maketitle


\begin{abstract}
    This paper presents a sampled-data framework for the safe navigation of controlled 
    agents in environments cluttered with obstacles governed by uncertain linear dynamics. 
    Collision-free motion is achieved by combining Control Barrier Function (CBF)-based safety 
    filtering with set-valued state estimation using Constrained Convex Generators (CCGs). At each 
    sampling time, a CCG estimate of each obstacle is obtained using a finite-horizon guaranteed 
    estimation scheme and propagated over the sampling interval to obtain a CCG-valued flow that 
    describes the estimated obstacle evolution. However, since CCGs are defined indirectly---as an 
    affine transformation of a generator set subject to equality constraints, rather than as a 
    sublevel set of a scalar function---converting the estimated obstacle flows into CBFs is a 
    nontrivial task. One of the main contributions of this paper is a procedure to perform this 
    conversion, ultimately yielding a CBF via a convex optimization problem whose validity is 
    established by the Implicit Function Theorem. The resulting obstacle-specific CBFs are then 
    merged into a single CBF that is used to design a safe controller through the standard 
    Quadratic Program (QP)-based approach. Since CCGs support Minkowski sums, the proposed 
    framework also naturally handles rigid-body agents and generalizes existing CBF-based 
    rigid-body navigation designs to arbitrary agent and obstacle geometries. While the main 
    contribution is general, the paper primarily focuses on agents with first-order control-affine 
    dynamics and second-order strict-feedback dynamics. Simulation examples demonstrate the 
    effectiveness of the proposed method.
\end{abstract}



\section{Introduction}

\IEEEPARstart{T}{he} collision-free motion of controlled agents in dynamic and uncertain 
environments is a fundamental requirement in many real-world applications. In cooperative 
scenarios, such as fleets of robots sharing a workspace or teams of autonomous vehicles 
coordinating maneuvers, agents may exchange information, yet remain subject to disturbances, 
measurement noise, and communication delays~\cite{garcia2020guaranteed}, 
\cite{rego2025cooperative}. Conversely, in noncooperative or adversarial settings, where the 
environment is cluttered with uncontrolled objects and independent autonomous agents, these 
entities follow their trajectories without disclosing state information or intentions to the 
navigating agent~\cite{bianchin2019secure}. Across these situations, safe navigation hinges on 
accurately capturing uncertainty, ideally while avoiding excessive conservatism.

In such contexts, guaranteed (set-valued) state estimation is a powerful alternative to stochastic 
filtering. Instead of relying on probabilistic assumptions, such as in Conditional Value-at-Risk 
(CVaR) approaches \cite{wang2025safe}, set-valued estimation algorithms compute sets that are 
guaranteed to contain the true state of a dynamical system, given bounds on the system inputs and 
the measurement noise \cite{combastel2016extended}. As highlighted in recent surveys, this 
framework has strong theoretical appeal and has sparked a lot of interest from the research 
community \cite{althoff2021comparison}, \cite{silvestre2024comparison}.

A key consideration when designing guaranteed state estimation algorithms is the set 
representation. For linear systems, numerous methods have been developed based on different set 
representations, evolving from intervals \cite{thabet2014effective}, \cite{tang2019interval} and 
ellipsoids \cite{kurzhanski2000ellipsoidal}, \cite{chernousko2005ellipsoidal} to more expressive 
structures such as zonotopes \cite{combastel2003state}, \cite{niazi2023resilient} and Constrained 
Zonotopes (CZs) \cite{scott2016constrained}, \cite{silvestre2017stochastic}, which have a reduced
wrapping effect compared to intervals and ellipsoids. For nonlinear systems, these techniques can 
still be employed, provided that the dynamics are approximated by a linear model to allow set 
propagation, as performed in \cite{alamo2005guaranteed}, \cite{abdallah2008box}, 
\cite{rego2020guaranteed}, and \cite{wan2018guaranteed}, for each of the mentioned set classes.

Most recently, Constrained Convex Generators (CCGs) have emerged as a unifying set representation 
capable of capturing highly general convex sets, subsuming all previously discussed set classes and 
greatly reducing the need for approximations. CCGs are sets described through an affine
transformation of a generator set in a (typically) higher-dimensional space, subject to linear 
equality constraints, enabling expressive shapes while supporting key operations such as affine 
mappings, Minkowski sums, and generalized intersections. Such features make CCGs particularly well 
suited for set-based estimation, enabling tight enclosures even when the system dynamics and 
measurement model induce mixtures of polytopic and ellipsoidal geometry. Consequently, CCGs 
currently represent the state of the art in convex set-valued estimation 
\cite{silvestre2021constrained}, \cite{silvestre2022accurate}, \cite{silvestre2022set}, 
\cite{silvestre2023exact}, \cite{silvestre2024closed}. 

It is important to note, however, that a central challenge in guaranteed state estimation is the 
growth of the underlying data structures over time, requiring the use of order reduction methods to 
maintain a fixed computational load. Several order reduction techniques have been developed for 
specific set representations, including zonotopes \cite{combastel2015zonotopes}, ellipsotopes 
\cite{kousik2022ellipsotopes}, CZs \cite{scott2016constrained}, and CCGs \cite{rego2025novel}. 
More recently, \cite{rego2024explicit} introduces a CCG finite-horizon scheme that refines an 
estimate computed by an ellipsoidal observer using a limited history of measurements, thereby 
removing the need for order reduction methods.

Combining Model Predictive Control (MPC) with CCG-based state estimation for collision-free motion 
in the presence of obstacles with uncertain dynamics is relatively straightforward, as safety 
constraints concerning estimated CCG obstacle sets can be directly derived from the definition of a 
CCG. Since a CCG is expressed as an affine map of a generator set subject to linear equality 
constraints, these conditions can be encoded into the MPC formulation to enforce obstacle 
avoidance, with the generator variable of the CCG included as an optimization variable. However, 
this approach leads to an MPC formulation with nonconvex constraints, an issue that becomes even 
more severe when the agent dynamics are nonlinear \cite{lindqvist2020nonlinear}. As a result, 
despite the predictive advantages of MPC, it adds a significant computational burden for real-time 
implementation, even when resorting to convexification techniques \cite{silvestre2023model}, 
\cite{mao2017successive}, \cite{taborda2024convex}.

Over the last decade, Control Barrier Functions (CBFs) have emerged as a powerful tool for 
designing safe controllers for nonlinear systems \cite{ames2016control}, \cite{ames2019control}, 
\cite{matias2025hybrid}. One of the primary utilities of CBFs is their ability to serve as a 
safety filter for a nominal controller, which may not have been designed to ensure safety. Such 
safety filters are typically instantiated through Quadratic Programs (QPs), which can be 
efficiently executed in real time, to minimize the deviation from the nominal controller while 
satisfying a Lyapunov-like condition that guarantees forward invariance of a designated safe set 
\cite{jankovic2018robust}, \cite{cohen2023characterizing}, \cite{alyaseen2025continuity}. However, 
in contrast to MPC, combining CBF-based safety filtering with CCG-based state estimation for safe 
navigation in the presence of obstacles with uncertain dynamics is a nontrivial task. Since CCGs 
describe sets indirectly---as an affine map of a generator set subject to equality constraints, 
rather than as a sublevel set of a scalar function---estimated CCG obstacle sets cannot be 
directly translated into CBFs. To the best of our knowledge, no existing work addresses the problem 
of integrating CBF-based control with guaranteed state estimation based on CCGs or related set 
representations such as CZs.

For deterministic environments, however, some CBF-based methods have been proposed for safe 
navigation under different agent and obstacle geometries. The main challenge tackled by this line 
of research is that ensuring safety for a rigid-body agent requires maintaining the entire agent 
set in a safe region, while classical CBF-based techniques only ensure safety for a single point. 
Examples include CBF-based methods for robotic arms avoiding point obstacles 
\cite{hamatani2020collision}, circular agents navigating around ellipses, cardioids, diamonds, and 
squares via discrete barrier states and differential dynamic programming 
\cite{almubarak2022safety}, and distributed multi-agent navigation based on circular agent and 
obstacle geometries \cite{mestres2024distributed}. Further examples include CBF-based techniques 
for manipulators interacting with humans modeled as capsules \cite{landi2019safety}, and 
manipulators moving under more general obstacle geometries using the signed distance function 
\cite{singletary2022safety}.   

Additionally, a few studies have explored CBF-based strategies for safe navigation under polytopic 
geometries. For example, \cite{long2024safe} navigates polygonal agents in elliptical environments 
based on the polygon-ellipse distance, \cite{thirugnanam2022safety} addresses polytope-polytope 
avoidance via discrete-time CBF constraints, and \cite{tayal2024polygonal} designs polygonal 
cone CBFs based on the vertices of polygon obstacles. More recently, \cite{chen2025control} 
proposes an optimization-based CBF that directly considers the exact signed distance function 
between a polytopic agent and polytopic obstacles, while \cite{molnar2025navigating} presents an 
optimization-free alternative for polytope-polytope navigation based on smooth maximum/minimum 
functions.

Despite recent progress, the aforementioned approaches are still quite specific, relying on 
constructions that do not readily extend to richer set classes, such as mixtures of polytopes and 
ellipsoids or sets defined by general $\ell_p$-norms. In this context, CCGs---beyond their 
original role in state estimation---provide a natural and unifying set representation for rigid 
body navigation. Since CCGs are closed under Minkowski sums, rigid-body agents can be equivalently 
treated as single points while the obstacles are enlarged with the agent geometry, enabling a 
systematic treatment of general agent and obstacle shapes once a mechanism for translating CCGs 
into CBFs is developed.


\subsection{Main Contributions and Organization}

The main goal of this paper is to introduce the theoretical framework for integrating CBF-based 
control with guaranteed state estimation based on CCGs, providing the foundations for further 
developments in this direction. To this end, we address the safe navigation of a controlled agent 
in an environment populated with obstacles governed by uncertain linear dynamics, with measurements 
obtained at discrete sampling instants.

At each sampling instant, a CCG estimate of each obstacle is computed using a finite-horizon 
guaranteed estimation scheme that generalizes the approach from \cite{rego2024explicit}. Each 
estimate is then propagated over the sampling interval to obtain a CCG-valued flow that describes 
the estimated obstacle evolution, being the basis for safety enforcement. The central challenge 
then lies in translating these CCG-valued flows into (time-varying) CBFs. The main contribution of 
the paper is a procedure that enables this conversion, which ultimately produces a CBF via a 
convex optimization problem whose validity is established through the Implicit Function Theorem 
under mild regularity conditions. The resulting obstacle-specific CBFs are then combined into a 
single CBF using a smooth approximation of the minimum function, and the overall CBF is employed 
to synthesize a safe controller via the standard QP-based approach. 

Since CCGs support Minkowski sums, the proposed framework naturally handles rigid-body agents, 
generalizing existing CBF-based designs for rigid-body navigation to arbitrary agent and obstacle 
geometries. Also, while the main contribution is general, the paper initially focuses on agents 
with first-order control-affine dynamics and then extends the method to agents with second-order 
strict-feedback dynamics via backstepping.

The remainder of the paper is structured as follows. Section \ref{Sec:Background} provides 
essential mathematical background, and Section \ref{Sec:Problem} formulates the safe 
navigation problem addressed in the paper. Sections \ref{Sec:Solution1}, \ref{Sec:Solution2}, and 
\ref{Sec:Solution3} describe the proposed approach, with simulation results in Section 
\ref{Sec:Results}. Finally, Section \ref{Sec:Conclusion} draws conclusions and outlines future 
directions.


\subsection{Notation and General Definitions}

$\mathbb{N}$ is the nonnegative integers set. $\mathbb{R}$, $\mathbb{R}_{\geq0}$, and 
$\mathbb{R}_{>0}$ are the sets of real, nonnegative, and positive numbers, respectively. 
$\mathbb{R}^n$ is the $n$-dimensional euclidean space, and $\mathbb{S}^{n-1}$ is the unit sphere 
in $\mathbb{R}^n$. $\mathbb{R}^{n\times m}$ is the set of $n\times m$ real matrices, 
$\mathbb{R}^{n\times n}_{\succ 0}$

\newpage

\noindent is the set of positive-definite matrices of size $n$, and 
$\mathrm{SO}(n)$ is the special orthogonal group in $\mathbb{R}^n$. For a set 
$\mathcal{S} \subseteq \mathbb{R}^n$, $\mathrm{int}(\mathcal{S})$ and $\partial\mathcal{S}$ are 
the interior and boundary of $\mathcal{S}$, respectively. The $p$-norm of a vector 
$\mathbf{x} \in \mathbb{R}^n$ is denoted $\|\mathbf{x}\|_p$ ($\|\mathbf{x}\| = \|\mathbf{x}\|_2$), 
and for two vectors $\mathbf{x}_1 \in\mathbb{R}^{n_1}$, $\mathbf{x}_2 \in \mathbb{R}^{n_2}$, we 
often use the notation $(\mathbf{x}_1, \mathbf{x}_2) = 
\left[\mathbf{x}_1^\top\, \mathbf{x}_2^\top\right]^\top \in \mathbb{R}^{n_1+n_2}$. For a 
differentiable function $h: \mathbb{R}^n\times[t_0, t_\mathrm{f}) \to \mathbb{R}$ and 
$\mathbf{G}: \mathbb{R}^n \to \mathbb{R}^{n \times m}$, we consider the Lie-derivative notation 
$L_\mathbf{G}h(\mathbf{x}, t) = 
\frac{\partial}{\partial\mathbf{x}}h(\mathbf{x}, t) \mathbf{G}(\mathbf{x})$, and 
$\Dot{\mathbf{x}}$ is the time derivative of $\mathbf{x} \in \mathbb{R}^n$. Additionally, 
$\mathbf{0}_{n\times m}$ is the $n\times m$ matrix of zeros, and $\mathbf{I}_{n}$ is the identity 
matrix of size $n$ (dimensions are often omitted when clear from context). The Cartesian product 
is represented by $\times$, the Minkowski sum by $\oplus$, and the generalized set intersection by 
$\cap_\mathbf{R}$. Finally, given the matrices $\mathbf{A}_1, \dots, \mathbf{A}_M$, 
$\mathrm{diag}(\mathbf{A}_1, \dots, \mathbf{A}_M)$ yields a block diagonal matrix, and 
$\mathrm{null}(\mathbf{A})$ yields a matrix whose columns form an orthonormal basis for the null 
space of a matrix $\mathbf{A}$.

\begin{definition}[Extended Class-$\mathcal{K}$/$\mathcal{K}_\infty$ Function]
    A continuous function $\alpha: \mathbb{R} \to \mathbb{R}$ is said to be an extended 
    class-$\mathcal{K}$ function if it is strictly increasing with $\alpha(0) = 0$, and it is an 
    extended class-$\mathcal{K}_\infty$ function if, additionally, 
    $\lim_{s \to \pm\infty}\alpha(s) = \pm\infty$.
\end{definition}

\begin{definition}[Graph of a Set-Valued Flow]
    Consider a set-valued flow $\mathcal{D}: [t_0, t_\mathrm{f}) \rightrightarrows \mathbb{R}^n$. 
    The set $\mathcal{G}(\mathcal{D})$, defined as
    \begin{equation}
        \mathcal{G}(\mathcal{D}) =
        \{(\mathbf{x}, t) \in \mathbb{R}^n\times[t_0, t_\mathrm{f}): 
        \mathbf{x} \in \mathcal{D}(t)\},
    \end{equation}
    is said to be the graph of $\mathcal{D}$.
\end{definition}


\section{Mathematical Background} \label{Sec:Background}

This section provides essential mathematical background on CBFs. Since we will deal with dynamic 
obstacle avoidance, we present the concepts in a time-varying fashion, as the standard static 
formulations naturally arise as particular cases.

Consider a nonlinear control-affine system of the form
\begin{equation}
    \dot{\mathbf{x}} = \mathbf{f}(\mathbf{x}) + \mathbf{G}(\mathbf{x})\mathbf{u},
    \label{Eq:ControlAffineSystem}
\end{equation}
where $\mathbf{x} \in \mathbb{R}^n$ is the system state, $\mathbf{u} \in \mathbb{R}^m$ is the 
control input, and the fields $\mathbf{f}: \mathbb{R}^n \to \mathbb{R}^n$ and 
$\mathbf{G}: \mathbb{R}^n \to \mathbb{R}^{n \times m}$ are locally Lipschitz. Let now 
$\mathbf{k}: \mathcal{G}(\mathcal{D}) \to \mathbb{R}^m$ be a feedback controller, defined over the 
graph of a set-valued flow $\mathcal{D}: [t_0, t_\mathrm{f}) \rightrightarrows \mathbb{R}^n$, for 
some given initial and final instants $t_0$ and $t_\mathrm{f}$ ($t_\mathrm{f}$ possibly infinite). 
Applying $\mathbf{k}$ to \eqref{Eq:ControlAffineSystem} yields the time-varying closed-loop system
\begin{equation}
    \Dot{\mathbf{x}} = \mathbf{f}(\mathbf{x}) + \mathbf{G}(\mathbf{x})\mathbf{k}(\mathbf{x}, t).
    \label{Eq:ControlAffineSystemClosedLoop}
\end{equation}
As the functions $\mathbf{f}$ and $\mathbf{G}$ are locally Lipschitz, if the controller 
$\mathbf{k}$ is locally Lipschitz in $\mathbf{x}$ and continuous in $t$, then, for every initial 
condition $\mathbf{x}_0 \in \mathcal{D}(t_0)$, there exists a unique continuously differentiable 
solution $\bm{\varphi}: I(\mathbf{x}_0) \to \mathbb{R}^n$ satisfying 
\begin{equation}
    \Dot{\bm{\varphi}}(t) = \mathbf{f}(\bm{\varphi}(t)) 
    + \mathbf{G}(\bm{\varphi}(t))\mathbf{k}(\bm{\varphi}(t), t),\quad
    \bm{\varphi}(t_0) = \mathbf{x}_0,
\end{equation}
for all $t \in I(\mathbf{x}_0)$, where $I(\mathbf{x}_0) \subseteq [t_0, t_\mathrm{f})$ denotes the 
maximal interval of existence for the solution \cite{perko2013differential}. In the remainder of 
this paper, we assume that $I(\mathbf{x}_0) = [t_0, t_\mathrm{f})$ for convenience. Next, we 
revisit the standard definition of forward invariance for sets and extend it to set-valued flows.

\begin{definition}[Forward Invariance (Set)]
    A set $\mathcal{C} \subset \mathbb{R}^n$ is said to be forward invariant for the system 
    \eqref{Eq:ControlAffineSystemClosedLoop} if, for every initial condition 
    $\mathbf{x}_0 \in \mathcal{C}$, we have $\bm{\varphi}(t) \in \mathcal{C}$ for all 
    $t \in [t_0, t_\mathrm{f})$.
\end{definition}

\newpage

\begin{definition}[Forward Invariance (Set-Valued Flow)]
    A set-valued flow $\mathcal{C}: [t_0, t_\mathrm{f}) \rightrightarrows \mathbb{R}^n$ is said to 
    be forward invariant for the system \eqref{Eq:ControlAffineSystemClosedLoop} if, for every 
    initial condition $\mathbf{x}_0 \in \mathcal{C}(t_0)$, we have 
    $\bm{\varphi}(t) \in \mathcal{C}(t)$ for all $t \in [t_0, t_\mathrm{f})$.
\end{definition}


\subsection{Control Barrier Functions}

We intend to ensure that, at every time instant, the solution of the closed-loop system 
\eqref{Eq:ControlAffineSystemClosedLoop} lies within a safe set 
$\mathcal{C}(t) \subset \mathbb{R}^n$, which amounts to ensuring forward invariance of a 
set-valued flow $\mathcal{C}: [t_0, t_\mathrm{f}) \rightrightarrows \mathbb{R}^n$ for system 
\eqref{Eq:ControlAffineSystemClosedLoop}. Particularly, we consider a set-valued flow 
$\mathcal{C}$ defined for all $t \in [t_0, t_\mathrm{f})$ as
\begin{equation}
    \mathcal{C}(t) = \{\mathbf{x} \in \mathbb{R}^n: h(\mathbf{x}, t) \geq 0\},
    \label{Eq:SafeFlow}
\end{equation}
where $h: \mathbb{R}^n\times[t_0, t_\mathrm{f}) \to \mathbb{R}$ is a continuously differentiable 
function with $\frac{\partial}{\partial\mathbf{x}}h(\mathbf{x}, t) \neq \mathbf{0}$ when 
$h(\mathbf{x}, t) = 0$. This regularity condition implies that, for all 
$t \in [t_0, t_\mathrm{f})$,
\begin{equation}
    \begin{aligned}
        \partial\mathcal{C}(t) &= \{\mathbf{x} \in \mathbb{R}^n: h(\mathbf{x}, t) = 0\},\\
        \mathrm{int}(\mathcal{C}(t)) &= \{\mathbf{x} \in \mathbb{R}^n: h(\mathbf{x}, t) > 0\}.
    \end{aligned}
\end{equation}
If $h$ has the properties of a CBF, then it can be used to design safe controllers for the system 
\eqref{Eq:ControlAffineSystem}.

\begin{definition}[CBF \cite{notomista2020persistification}] \label{Def:CBF}
    Let $\mathcal{C}: [t_0, t_\mathrm{f}) \rightrightarrows \mathbb{R}^n$ be a set-valued flow 
    defined by \eqref{Eq:SafeFlow}, for a continuously differentiable function 
    $h: \mathbb{R}^n\times[t_0, t_\mathrm{f}) \to \mathbb{R}$ with 
    $\frac{\partial}{\partial\mathbf{x}}h(\mathbf{x}, t) \neq \mathbf{0}$ when 
    $h(\mathbf{x}, t) = 0$. The function $h$ is a (zeroing) CBF for the system 
    \eqref{Eq:ControlAffineSystem} if there exists a set-valued flow 
    $\mathcal{D}: [t_0, t_\mathrm{f}) \rightrightarrows \mathbb{R}^n$ with 
    $\mathcal{C}(t) \subseteq \mathcal{D}(t)$ for all $t \in [t_0, t_\mathrm{f})$ and an extended 
    class-$\mathcal{K}_\infty$ function $\alpha: \mathbb{R} \to \mathbb{R}$ such that, for all 
    $(\mathbf{x}, t) \in \mathcal{G}(\mathcal{D})$,
    \begin{equation}
        \sup_{\mathbf{u} \in \mathbb{R}^m}
        \Dot{h}(\mathbf{x}, t, \mathbf{u}) > - \alpha(h(\mathbf{x}, t)),
        \label{Eq:DefinitionCBF}
    \end{equation}
    where the function $\Dot{h}$ is defined as
    \begin{equation}
        \Dot{h}(\mathbf{x}, t, \mathbf{u}) = 
        L_\mathbf{f}h(\mathbf{x}, t) + L_\mathbf{G}h(\mathbf{x}, t)\mathbf{u} 
        + \frac{\partial h(\mathbf{x}, t)}{\partial t}.
    \end{equation}
\end{definition}

\vspace{1mm}

Given a CBF $h$ for \eqref{Eq:ControlAffineSystem} and an associated extended 
class-$\mathcal{K}_\infty$ function $\alpha$, we define the pointwise set of controls
\begin{equation}
    K_\mathrm{CBF}(\mathbf{x}, t) = \left\{\mathbf{u} \in \mathbb{R}^m: 
    \Dot{h}(\mathbf{x}, t, \mathbf{u}) \geq -\alpha(h(\mathbf{x}, t))\right\}.
    \label{Eq:SetCBF}
\end{equation}
This yields the following main result concerning CBFs.

\begin{lemma}[Safe Controller \cite{notomista2020persistification}]
    Let $\mathcal{C}: [t_0, t_\mathrm{f}) \rightrightarrows \mathbb{R}^n$ be a set-valued flow 
    defined by \eqref{Eq:SafeFlow}, for a continuously differentiable function 
    $h: \mathbb{R}^n\times[t_0, t_\mathrm{f}) \to \mathbb{R}$ such that 
    $\frac{\partial}{\partial\mathbf{x}}h(\mathbf{x}, t) \neq \mathbf{0}$ when 
    $h(\mathbf{x}, t) = 0$. If the function $h$ is a CBF for the system 
    \eqref{Eq:ControlAffineSystem} on a set-valued flow 
    $\mathcal{D}: [t_0, t_\mathrm{f}) \rightrightarrows \mathbb{R}^n$, then the set 
    $K_\mathrm{CBF}(\mathbf{x}, t)$ is nonempty for all 
    $(\mathbf{x}, t) \in \mathcal{G}(\mathcal{D})$, and every feedback controller 
    $\mathbf{k}: \mathcal{G}(\mathcal{D}) \to \mathbb{R}^m$ that is locally Lipschitz in 
    $\mathbf{x}$, continuous in $t$, and with 
    $\mathbf{k}(\mathbf{x}, t) \in K_\mathrm{CBF}(\mathbf{x}, t)$ for all 
    $(\mathbf{x}, t) \in \mathcal{G}(\mathcal{D})$ renders $\mathcal{C}$ forward invariant for the 
    resulting closed-loop system.
\end{lemma}

\begin{remark}
    The strictness of the inequality \eqref{Eq:DefinitionCBF} enables proving that 
    optimization-based controllers relying on CBFs are locally Lipschitz continuous in 
    $\mathbf{x}$ \cite{morris2013sufficient}, \cite{jankovic2018robust}.
\end{remark}


\subsection{Safety Filters}

One of the major utilities of CBFs is their ability to serve as a safety filter for a nominal 
controller $\mathbf{k}_\mathrm{d}: \mathbb{R}^n\times [t_0, t_\mathrm{f}) \to \mathbb{R}^m$. A 
safety filter is a controller that modifies $\mathbf{k}_\mathrm{d}$, preferably in a minimally 
invasive fashion, such that the resulting closed-loop 

\newpage

\noindent system is safe. Given a CBF $h: \mathbb{R}^n\times[t_0, t_\mathrm{f}) \to \mathbb{R}$ 
for \eqref{Eq:ControlAffineSystem} on 
$\mathcal{D}: [t_0, t_\mathrm{f}) \rightrightarrows \mathbb{R}^n$, the typical approach for 
designing a safety filter $\mathbf{k}: \mathcal{G}(\mathcal{D}) \to \mathbb{R}^m$ is through the 
following QP: 
\begin{equation} 
    \begin{aligned}
        \mathbf{k}(\mathbf{x}, t) = \underset{\mathbf{u} \in \mathbb{R}^m}{\arg\min}\,\,
        &\tfrac{1}{2}\|\mathbf{u} - \mathbf{k}_\mathrm{d}(\mathbf{x}, t)\|^2\\
        \mathrm{s.t.}\,\, &\Dot{h}(\mathbf{x}, t, \mathbf{u}) \geq -\alpha(h(\mathbf{x}, t)),
    \end{aligned}
    \label{Eq:SafetyFilter}
\end{equation}
where $\alpha$ is an extended class-$\mathcal{K}_\infty$ function associated with the CBF. 
As per the Karush-Kuhn-Tucker (KKT) conditions, the controller defined in \eqref{Eq:SafetyFilter} 
can be expressed in closed form as
\begin{equation}
    \mathbf{k}(\mathbf{x}, t) = \mathbf{k}_\mathrm{d}(\mathbf{x}, t) 
    + \mu(\mathbf{x}, t)L_\mathbf{G}h(\mathbf{x}, t)^\top,
\end{equation}
where $\mu(\mathbf{x}, t)$ is the KKT multiplier associated with the CBF constraint, defined for 
all $(\mathbf{x}, t) \in \mathcal{G}(\mathcal{D})$ as
\begin{equation}
    \hspace{-1mm}
    \mu(\mathbf{x}, t) =
    \begin{cases}
        \bar{\mu}(\mathbf{x}, t), &\hspace{-2mm}\text{if }
        \Dot{h}(\mathbf{x}, t, \mathbf{k}_\mathrm{d}(\mathbf{x}, t)) 
        \leq -\alpha(h(\mathbf{x}, t)),\\
        0, &\hspace{-2mm}\text{if } \Dot{h}(\mathbf{x}, t, \mathbf{k}_\mathrm{d}(\mathbf{x}, t)) 
        > -\alpha(h(\mathbf{x}, t)),
    \end{cases}
\end{equation}
where the expression corresponding to the first case is
\begin{equation}
    \bar{\mu}(\mathbf{x}, t) = 
    -\frac{\Dot{h}(\mathbf{x}, t, \mathbf{k}_\mathrm{d}(\mathbf{x}, t)) + \alpha(h(\mathbf{x}, t))}
    {\|L_\mathbf{G}h(\mathbf{x}, t)\|^2}.
\end{equation}
Also, the controller defined in \eqref{Eq:SafetyFilter} is locally Lipschitz in $\mathbf{x}$ 
and continuous in $t$, provided the CBF gradient is locally Lipschitz in $\mathbf{x}$, $\alpha$ is 
locally Lipschitz, and the nominal controller $\mathbf{k}_\mathrm{d}$ is locally Lipschitz in 
$\mathbf{x}$ and continuous in $t$ \cite{morris2013sufficient}.


\section{Problem Formulation} \label{Sec:Problem}

We consider a rigid-body agent whose configuration at any time instant is described by the set
\begin{equation}
    \mathcal{P}(\mathbf{p}) = \mathbf{p} + \bar{\mathcal{P}},
    \label{Eq:AgentSet}
\end{equation}
where $\mathbf{p} \in \mathbb{R}^p$ is the position of the agent, matching its center, and 
$\bar{\mathcal{P}} \subset \mathbb{R}^p$ is a compact convex set with nonempty interior that 
defines the shape of the agent relative to its center. Also, we consider two types of dynamics for 
the agent, with different relative degrees. In the first scenario, the agent dynamics are described 
by a first-order control-affine model of the form
\begin{equation}
    \Dot{\mathbf{p}} = \mathbf{f}(\mathbf{p}) + \mathbf{G}(\mathbf{p})\mathbf{z},
    \label{Eq:AgentDynamics1}
\end{equation}
where the state matches the position of the agent, $\mathbf{z} \in \mathbb{R}^z$ is the control 
input, the functions $\mathbf{f}: \mathbb{R}^p \to \mathbb{R}^p$ and 
$\mathbf{G}: \mathbb{R}^p \to \mathbb{R}^{p\times z}$ are locally Lipschitz, and the matrix 
$\mathbf{G}(\mathbf{p})$ has full row rank for all $\mathbf{p} \in \mathbb{R}^p$. In the second 
scenario, the agent dynamics are extended to a second-order strict-feedback model of the form
\begin{equation}
    \begin{aligned}
        \Dot{\mathbf{p}} &= \mathbf{f}(\mathbf{p}) + \mathbf{G}(\mathbf{p})\mathbf{z},\\
        \Dot{\mathbf{z}} &= \mathbf{f}_1(\mathbf{p}, \mathbf{z}) 
        + \mathbf{G}_1(\mathbf{p}, \mathbf{z})\mathbf{u},
    \end{aligned}
    \label{Eq:AgentDynamics2}
\end{equation}
where the state is now the pair $(\mathbf{p}, \mathbf{z}) \in \mathbb{R}^p\times\mathbb{R}^z$, 
$\mathbf{u} \in \mathbb{R}^m$ is the control input, the additional fields 
$\mathbf{f}_1: \mathbb{R}^p\times\mathbb{R}^z \to \mathbb{R}^z$ and 
$\mathbf{G}_1: \mathbb{R}^p\times\mathbb{R}^z \to \mathbb{R}^{z\times m}$ are also locally 
Lipschitz, and the gain matrix $\mathbf{G}_1(\mathbf{p}, \mathbf{z})$ has full row rank for all 
$(\mathbf{p}, \mathbf{z}) \in \mathbb{R}^p\times\mathbb{R}^z$.

The agent will operate in an environment cluttered with $M$ dynamic obstacles, indexed by 
$i \in \mathcal{I} = \{1, \dots, M\}$. For each $i \in \mathcal{I}$, the obstacle $i$ is 
represented by the set
\begin{equation}
    \mathcal{O}_i(\mathbf{o}_i) = \mathbf{o}_i + \bar{\mathcal{O}}_i,
\end{equation}
where $\mathbf{o}_i \in \mathbb{R}^p$ is the position of the obstacle, which matches its center, 
and $\bar{\mathcal{O}}_i \subset \mathbb{R}^p$ is a compact convex set with nonempty interior that 
defines the geometry of the obstacle relative to its center. From the standpoint of the agent, the 
obstacles evolve 

\newpage

\noindent according to uncertain linear dynamics. More specifically, the motion of obstacle $i$ is 
governed by the linear model
\begin{equation}
    \begin{aligned}
        \Dot{\mathbf{x}}_i &= \mathbf{F}_i\mathbf{x}_i + \mathbf{w}_i,\\
        \mathbf{o}_i &= \mathbf{E}_i\mathbf{x}_i,
    \end{aligned}
    \label{Eq:ObstacleDynamics}
\end{equation}
where $\mathbf{x}_i \in \mathbb{R}^{n_i}$ is the state, $\mathbf{w}_i \in \mathbb{R}^{n_i}$ 
represents an uncertain input, $\mathbf{F}_i \in \mathbb{R}^{n_i\times n_i}$ is a known matrix 
that describes the free dynamics, and $\mathbf{E}_i \in \mathbb{R}^{p\times n_i}$ is a known matrix 
that extracts the position from the state. The initial state $\mathbf{x}_{i, 0}$ and input 
$\mathbf{w}_i$ are known to belong to compact convex sets 
$\mathcal{X}_{i, 0}, \mathcal{W}_i \subset \mathbb{R}^{n_i}$.

At each sampling instant $t_k = kT_s$, for $k \in \mathbb{N}$ and sampling period 
$T_s \in \mathbb{R}_{>0}$, the agent obtains a measurement 
$\mathbf{y}_{i, k} \in \mathbb{R}^{y_i}$ of the state of each obstacle $i \in \mathcal{I}$. Each 
measurement follows a linear observation model of the form
\begin{equation}
    \mathbf{y}_{i, k} = \mathbf{C}_i\mathbf{x}_{i, k} + \mathbf{v}_{i, k},
    \label{Eq:MeasurementModel}
\end{equation}
where $\mathbf{C}_i \in \mathbb{R}^{y_i\times n_i}$ is the observation matrix, and 
$\mathbf{v}_{i, k} \in \mathbb{R}^{y_i}$ is the measurement noise, which belongs to a compact 
convex set $\mathcal{V}_i \subset \mathbb{R}^{y_i}$. Also, $\mathbf{x}_{i, k}$ denotes the state 
of the $i$th obstacle at the sampling instant $t_k$. With this setup in place, we are now ready to 
formally state the problem addressed in this paper.

\begin{problem}
    Building on the previous setup, consider also a nominal controller for the agent, defined for 
    all $t \in \mathbb{R}_{\geq0}$, that encodes a desired control objective. Then, design a 
    control algorithm that modifies the nominal controller, preferably in a minimally invasive way, 
    to ensure that the agent never collides with any obstacle. Formally, the algorithm must 
    guarantee that
    \begin{equation}
        \mathcal{P}(\bm{\varphi}(t)) 
        \cap \left(\bigcup_{i \in \mathcal{\mathcal{I}}}\mathcal{O}_i(\bm{\varphi}_i(t))\right) 
        = \emptyset
        \label{Eq:ObstacleAvoidance}
    \end{equation}
    for all $t \in \mathbb{R}_{\geq0}$, where $\bm{\varphi}: \mathbb{R}_{\geq0} \to \mathbb{R}^p$ 
    is the trajectory of the agent, and $\bm{\varphi}_i: \mathbb{R}_{\geq0} \to \mathbb{R}^p$ is 
    the trajectory of obstacle $i$.
\end{problem}

\begin{remark}[Practical Scenarios]
    This problem setup captures several practical scenarios of interest:
    \begin{itemize}
        \item Cooperative agents: an obstacle may correspond to another autonomous agent that 
        shares information, such as its control law and estimates. In such cases, $\mathbf{F}_i$ 
        represents the resulting closed-loop dynamics, while $\mathbf{w}_i$ accounts for bounded 
        disturbances and modeling errors.
        \item Noncooperative agents and uncontrolled objects: Additionally, an obstacle may 
        correspond to a noncooperative agent or uncontrolled object for which only input bounds are 
        known. In such cases, the control input and disturbances are absorbed into the uncertain 
        input $\mathbf{w}_i$, i.e., 
        \begin{equation}
            \mathbf{w}_i = \mathbf{B}_i\mathbf{u}_i + \mathbf{d}_i,
        \end{equation}
        where $\mathbf{u}_i$ is a control input known to belong to a compact convex set 
        $\mathcal{U}_i$ and $\mathbf{d}_i$ is a disturbance that belongs to a compact convex set 
        $\mathcal{D}_i$, and consequently 
        \begin{equation}
            \mathcal{W}_i = \mathbf{B}_i\,\mathcal{U}_i \oplus \mathcal{D}_i.
        \end{equation}
    \end{itemize}
\end{remark}

\begin{remark}[Agent Geometry]
    Note that we consider that the agent moves without rotation, so its attitude is not modeled and 
    the body set $\bar{\mathcal{P}}$ is constant. This assumption is made only for clarity of 
    exposition and because agent rotation is not central to the focus of this paper. The approach 
    proposed herein readily extends to rotating agents, in which case the occupied set in

    \newpage
    
    \noindent \eqref{Eq:AgentSet} is defined as  
    $\mathcal{P}(\mathbf{p},\bm{\psi}) = \mathbf{p} + \mathbf{R}(\bm{\psi})\bar{\mathcal{P}}$, 
    where $\bm{\psi} \in \mathbb{R}^\psi$ is an attitude representation, 
    $\mathbf{R}(\bm{\psi}) \in \mathrm{SO}(p)$ is a rotation matrix, and $\bar{\mathcal{P}}$ is 
    the body set for $\mathbf{R}(\bm{\psi}) = \mathbf{I}$. We also note that $\mathbf{p}$ does 
    not need to be the physical center of the agent; it may be any fixed point on its body. Under 
    a coordinate change $\mathbf{p}' = \mathbf{p} + \mathbf{t}$, the body set of the agent becomes 
    $\bar{\mathcal{P}}' = \bar{\mathcal{P}} - \mathbf{t}$.
\end{remark}

\begin{remark}[Agent Dynamics]
    The full row rank of the gain matrices in \eqref{Eq:AgentDynamics1} and 
    \eqref{Eq:AgentDynamics2} means that the agent can be stabilized to any position. Nevertheless, 
    we highlight that these models are quite general. Even if the dynamics do not originally match 
    any of these forms, they can often be transformed to do so via a suitable change of 
    coordinates. For instance, in the case of the unicycle model, controlling a point slightly 
    displaced from the original control point yields a first-order control-affine system with a 
    full-row-rank gain matrix \cite{glotfelter2019hybrid}. This trick is commonly used for 
    controlling underactuated vehicles \cite{aguiar2007trajectory}, \cite{reis2022nonlinear}. Also, 
    we emphasize that these models capture most cases of practical interest, such as ground, 
    marine, and aerial vehicles \cite{jiangdagger1997tracking}, \cite{yang2004combined}, 
    \cite{kim2004robust}, spacecraft \cite{farrell2005backstepping}, \cite{sun2016adaptive}, and 
    others. Despite the generality of the systems under consideration, we emphasize that these 
    serve only to allow a more systematic control design and are not directly related to the main 
    contributions of this paper.
\end{remark}


\section{Proposed Solution} \label{Sec:Solution1}

This section provides a high-level description of our solution approach, which relies on two main 
components: a guaranteed state estimation algorithm based on CCGs (Section \ref{Sec:Solution2}), 
and a conversion procedure from CCGs to CBFs (Section \ref{Sec:Solution3}).

For convenience, we begin by observing that we can equivalently treat the agent as a point-mass 
agent and account for its body by enlarging the obstacle sets accordingly. Specifically, for each 
$i \in \mathcal{I}$, we introduce the enlarged obstacle set 
\begin{equation}
    \mathcal{O}^+_i(\mathbf{o}_i) = \mathbf{o}_i + \bar{\mathcal{O}}^+_i 
    = \mathbf{o}_i + \bar{\mathcal{O}}_i \oplus \left(-\bar{\mathcal{P}}\right),
\end{equation}
which incorporates the agent body. With such a definition, the obstacle avoidance condition can be 
equivalently expressed as
\begin{equation}
    \bm{\varphi}(t)
    \notin \bigcup_{i \in \mathcal{\mathcal{I}}}\mathcal{O}^+_i(\bm{\varphi}_i(t))
\end{equation}
for all $t \in \mathbb{R}_{\geq0}$. From now on, we adopt this equivalent point-mass standpoint 
for describing the proposed solution.

Our approach builds upon a state-of-the-art guaranteed state estimation algorithm based on CCGs, 
whose implementation is detailed in Section \ref{Sec:Solution2}. This algorithm is used for 
obtaining set-valued estimates of the obstacle states, which are guaranteed to contain the actual 
states. Specifically, at each sampling instant $t_k$, we apply this algorithm to each obstacle 
$i \in \mathcal{I}$, obtaining a CCG state estimate 
$\hat{\mathcal{X}}_{i, k} \subset \mathbb{R}^{n_i}$ that contains the actual state 
$\mathbf{x}_{i, k}$, based on the measurements collected up to that time.

During the sampling interval $[t_k, t_{k+1})$, we can then obtain guaranteed estimates of the 
obstacle sets by propagating the obstacle-position estimates through the obstacle dynamics and 
inflating them by the obstacle bodies. To capture the estimates of the $i$th obstacle during the 
sampling interval, we introduce the set-valued flow 
$\hat{\mathcal{O}}^+_{i, k}: [t_k, t_{k+1}) \rightrightarrows \mathbb{R}^p$, which assigns to each 
time $t \in [t_k, t_{k+1})$ an estimated obstacle set $\hat{\mathcal{O}}^+_{i, k}(t)$ that 
includes the true obstacle set. Such an estimate is constructed by projecting the reachable set of 
states onto the position space 

\newpage

\noindent and adding the obstacle body set as follows:
\begin{equation}
    \hat{\mathcal{O}}^+_{i, k}(t) = \mathbf{E}_i
    \left(\mathbf{\Phi}_i(t - t_k)\hat{\mathcal{X}}_{i, k} 
    \oplus \mathbf{\Gamma}_i(t - t_k)\tilde{\mathcal{W}}_i\right)
    \oplus \bar{\mathcal{O}}^+_i,
\end{equation}
where $\mathbf{\Phi}_i(s)$ denotes the state-transition matrix, given by
\begin{equation}
    \mathbf{\Phi}_i(s) = \exp(\mathbf{F}_is).
\end{equation}
Also, if the input $\mathbf{w}_i$ is known to be constant over the sampling interval, the function 
$\mathbf{\Gamma}_i: \mathbb{R} \to \mathbb{R}^{n_i\times n_i}$ is given exactly by
\begin{equation}
    \mathbf{\Gamma}_i(s) = \int_0^s\mathbf{\Phi}_i(\tau)d\tau,
\end{equation}
and the set $\tilde{\mathcal{W}}_i$ matches $\mathcal{W}_i$. If the evolution of the input over 
the sampling interval is not known, we can construct an over-approximation of the reachable set by 
defining $\mathbf{\Gamma}_i$ as
\begin{equation}
    \mathbf{\Gamma}_i(s) = \frac{\exp(\|\mathbf{F}_i\|s) - 1}{\|\mathbf{F}_i\|},
\end{equation}
in which case $\tilde{\mathcal{W}}_i$ is defined as a unit 2-norm ball centered at the origin with 
radius $\bar{w} = \max_{\mathbf{w}_i \in \mathcal{W}_i} \|\mathbf{w}_i\|$ (see, e.g., 
\cite{girard2005reachability}).

For each $i \in \mathcal{I}$, we then convert the set-valued flow $\hat{\mathcal{O}}^+_{i, k}$ 
into a CBF $h_{i, k}: \mathbb{R}^p\times[t_k, t_{k+1}) \to \mathbb{R}$ for 
\eqref{Eq:AgentDynamics1}, which defines a safe set-valued flow 
$\mathcal{C}_{i, k}: [t_k, t_{k+1}) \rightrightarrows \mathbb{R}^p$ for the agent as
\begin{equation}
    \mathcal{C}_{i, k}(t) = \{\mathbf{p} \in \mathbb{R}^p: h_{i, k}(\mathbf{p}, t) \geq 0\},
\end{equation}
such that $\mathcal{C}_{i, k}(t) \subseteq \mathbb{R}^p\setminus\hat{\mathcal{O}}^+_{i, k}(t)$ for 
all $t \in [t_k, t_{k+1})$. However, the challenge here is that the estimated obstacle set 
$\hat{\mathcal{O}}^+_{i, k}(t)$ will be a CCG, rendering the conversion nontrivial. The proposed 
conversion procedure is presented in Section \ref{Sec:Solution3}.

At each time $t \in [t_k, t_{k+1})$, the overall safe set for the agent position is the 
intersection of the individual safe sets,
\begin{equation}
    \bigcap_{i \in \mathcal{I}}\mathcal{C}_{i, k}(t) = \{\mathbf{p} \in \mathbb{R}^p:
    h_{i, k}(\mathbf{p}, t) \geq 0 \text{ for all } i \in \mathcal{I}\},
\end{equation}
which can be compactly written using the minimum function:
\begin{equation}
    \bigcap_{i \in \mathcal{I}}\mathcal{C}_{i, k}(t) = \left\{\mathbf{p} \in \mathbb{R}^p:
    \min_{i \in \mathcal{\mathcal{I}}} h_{i, k}(\mathbf{p}, t) \geq 0\right\}.
\end{equation}
However, as the minimum function is not differentiable, it cannot be directly used to define a 
CBF candidate. To address this, we consider a smooth approximation of the minimum function, 
adopting the LogSumExp approach from \cite{molnar2023composing}. Specifically, we define an 
overall safe set-valued flow $\mathcal{C}_k: [t_k, t_{k+1}) \rightrightarrows \mathbb{R}^p$ as
\begin{equation}
    \mathcal{C}_k(t) = \{\mathbf{p} \in \mathbb{R}^p: h_k(\mathbf{p}, t) \geq 0\},
\end{equation}
where the overall CBF candidate $h_k: \mathbb{R}^p\times[t_k, t_{k+1}) \to \mathbb{R}$ is defined 
for all $(\mathbf{p}, t) \in \mathbb{R}^p\times[t_k, t_{k+1})$ as
\begin{equation}
    h_k(\mathbf{p}, t) = -\frac{1}{\beta_k}
    \ln\!\left(\sum_{i \in \mathcal{I}}\exp(-\beta_k h_{i, k}(\mathbf{p}, t))\right) 
    - \frac{1}{\beta_k},
    \label{Eq:OverallCBF}
\end{equation}
with a smoothing parameter $\beta_k \in \mathbb{R}_{>0}$. Following \cite{molnar2023composing}, 
this construction ensures that, for all $(\mathbf{p}, t) \in \mathbb{R}^p\times[t_k, t_{k+1})$,
\begin{equation}
    0 < \min_{i \in \mathcal{\mathcal{I}}} h_{i, k}(\mathbf{p}, t) - h_k(\mathbf{p}, t) 
    < \frac{\ln(M+1) + 1}{\beta_k},
    \label{Eq:SafetyDetail}
\end{equation}
so that $\mathcal{C}_k(t) \subset \mathrm{int}(\cap_{i \in \mathcal{I}}\mathcal{C}_{i, k}(t))$, 
and, as $\beta_k \to \infty$, the smooth approximation converges to the exact intersection:
\begin{equation}
    \lim_{\beta_k\to\infty} \mathcal{C}_k(t) = 
    \bigcap_{i \in \mathcal{I}}\mathcal{C}_{i, k}(t).
\end{equation}
In particular, \eqref{Eq:SafetyDetail} shows that the approximation error decays 

\newpage

\noindent linearly with $1/\beta_k$, so $\beta_k$ can be chosen to guarantee a desired 
conservatism tolerance $\varepsilon_\beta$ by setting 
$\beta_k \geq (\ln(M+1)+1)/\varepsilon_\beta$. Finally, we note that the gradient of $h_k$ can be 
expressed as 
\begin{equation}
    \frac{\partial h_k(\mathbf{p}, t)}{\partial(\mathbf{p}, t)} = \sum_{i \in \mathcal{I}}
    \pi_i(\mathbf{p}, t)\frac{\partial h_{i, k}(\mathbf{p}, t)}{\partial(\mathbf{p}, t)},
\end{equation}
where each weight $\pi_i(\mathbf{p}, t)$ is defined as
\begin{equation}
    \pi_i(\mathbf{p}, t) = \exp(-\beta_k(h_{i, k}(\mathbf{p}, t) - h_k(\mathbf{p}, t)) + 1).
\end{equation}
If $h_k$ is a valid CBF for \eqref{Eq:AgentDynamics1}, then it can be used to design a safe 
controller for the sampling interval $[t_k, t_{k+1})$.

\begin{proposition}[Valid CBF for \eqref{Eq:AgentDynamics1}]
    Let $h: \mathbb{R}^p\times[t_0, t_\mathrm{f}) \to \mathbb{R}$ be a continuously differentiable 
    function with $\frac{\partial}{\partial\mathbf{p}}h(\mathbf{p}, t) \neq \mathbf{0}$ when 
    $h(\mathbf{p}, t) = 0$. If when 
    $\frac{\partial}{\partial\mathbf{p}}h(\mathbf{p}, t) = \mathbf{0}$ and $h(\mathbf{p}, t) > 0$,
    \begin{equation}
        \frac{\partial h(\mathbf{p}, t)}{\partial t} \geq 0,
        \label{Eq:ValidCBFAssumption}
    \end{equation}
    then $h$ is a CBF for the system \eqref{Eq:AgentDynamics1}, for an arbitrary associated 
    extended class-$\mathcal{K}_\infty$ function $\alpha$.
\end{proposition}

\begin{proof}
    Let $\mathcal{C}: [t_0, t_\mathrm{f}) \rightrightarrows \mathbb{R}^p$ be the set-valued flow 
    associated with $h$, defined for all $t \in [t_0, t_\mathrm{f})$ as
    \begin{equation}
        \mathcal{C}(t) = \{\mathbf{p} \in \mathbb{R}^p: h(\mathbf{p}, t) \geq 0\}.
    \end{equation}
    Given the differentiability and regularity assumptions, showing that $h$ is a CBF for 
    \eqref{Eq:AgentDynamics1} amounts to verifying the existence of a set-valued flow 
    $\mathcal{D}: [t_0, t_\mathrm{f}) \rightrightarrows \mathbb{R}^p$ with 
    $\mathcal{C}(t) \subseteq \mathcal{D}(t)$ for all $t \in [t_0, t_\mathrm{f})$ and an extended 
    class-$\mathcal{K}_\infty$ function $\alpha: \mathbb{R} \to \mathbb{R}$ so that, for all 
    $(\mathbf{p}, t) \in \mathcal{G}(\mathcal{D})$, we have
    \begin{equation}
        \sup_{\mathbf{z} \in \mathbb{R}^z} 
        \Dot{h}(\mathbf{p}, t, \mathbf{z}) > -\alpha(h(\mathbf{p}, t)).
        \label{Eq:DefinitionCBFPosition}
    \end{equation}

    We begin by noting that when 
    $\frac{\partial}{\partial\mathbf{p}}h(\mathbf{p}, t) \neq \mathbf{0}$, the value of $\Dot{h}$ 
    can be made arbitrarily large through input, meaning that
    \begin{equation}
        \sup_{\mathbf{z} \in \mathbb{R}^z} 
        \Dot{h}(\mathbf{p}, t, \mathbf{z}) = \infty > -\alpha(h(\mathbf{p}, t)),
    \end{equation}
    for any function $\alpha$. When 
    $\frac{\partial}{\partial\mathbf{p}}h(\mathbf{p}, t) = \mathbf{0}$ and $h(\mathbf{p}, t) > 0$,
    \begin{equation}
        \Dot{h}(\mathbf{p}, t, \mathbf{z}) = \frac{\partial h(\mathbf{p}, t)}{\partial t} \geq 0 
        > -\alpha(h(\mathbf{p}, t)),
    \end{equation}
    for any extended class-$\mathcal{K}_\infty$ function $\alpha$. Hence, for an arbitrary 
    extended class-$\mathcal{K}_\infty$ function $\alpha$, the inequality in 
    \eqref{Eq:DefinitionCBFPosition} holds for at least all 
    $(\mathbf{p}, t) \in \mathcal{G}(\mathcal{C})$, completing the proof. 
\end{proof}

\begin{remark}
    If the inequality in \eqref{Eq:ValidCBFAssumption} does not hold at all points where 
    $\frac{\partial}{\partial\mathbf{p}}h(\mathbf{p}, t) = \mathbf{0}$ and $h(\mathbf{p}, t) > 0$, 
    then we must find a particular extended class-$\mathcal{K}_\infty$ function $\alpha$ such that
    \begin{equation}
        \frac{\partial h(\mathbf{p}, t)}{\partial t} > -\alpha(h(\mathbf{p}, t))
    \end{equation}
    whenever $\frac{\partial}{\partial\mathbf{p}}h(\mathbf{p}, t) = \mathbf{0}$, 
    $h(\mathbf{p}, t) > 0$, and $\frac{\partial}{\partial t} h(\mathbf{p}, t) < 0$. In practice, 
    however, such critical points are rare. For a function $h$ defined as in \eqref{Eq:OverallCBF}, 
    if each individual CBF is analytic, then $h$ is analytic. Since a nonconstant analytic function 
    cannot have its gradient vanish identically on an open set, the set of positions $\mathbf{p}$ 
    where $\frac{\partial}{\partial\mathbf{p}}h(\mathbf{p}, t) = \mathbf{0}$ has Lebesgue measure 
    zero (see, e.g., \cite{krantz2002primer}), and for compact obstacles, as considered throughout 
    this paper, typically reduces to a collection of isolated positions. At such positions, 
    $\alpha$ can simply be adjusted locally to guarantee that the required condition holds, with no 
    meaningful effect.
\end{remark}

\newpage

\begin{algorithm}[H]
    \caption{Safe Navigation Scheme}
    \label{Alg:Solution}
    \begin{algorithmic}[1]
        \Require agent, obstacles, and measurement models
        \For{$k \in \mathbb{N}$}
            \For{$i \in \mathcal{I}$}
                \State collect measurement $\mathbf{y}_{i, k}$
                \State compute state estimate $\hat{\mathcal{X}}_{i, k}$ 
                (Section \ref{Sec:Solution2})
                \State convert $\hat{\mathcal{O}}_{i, k}$ into a CBF $h_{i, k}$ for 
                \eqref{Eq:AgentDynamics1} (Section \ref{Sec:Solution3})
            \EndFor
            \State design overall CBF $h_k$ / $h_{k, 1}$ for 
            \eqref{Eq:AgentDynamics1} / \eqref{Eq:AgentDynamics2}
            \State apply safe controller during $[t_k, t_{k+1})$
        \EndFor
    \end{algorithmic}
\end{algorithm}

When our agent has first-order dynamics as in \eqref{Eq:AgentDynamics1} and $h_k$ is a valid CBF 
for it on a set-valued flow $\mathcal{D}_k: [t_k, t_{k+1}) \rightrightarrows \mathbb{R}^p$, we can 
then construct a safety filter $\mathbf{k}_k: \mathcal{G}(\mathcal{D}_k) \to \mathbb{R}^z$ for the 
sampling interval $[t_k, t_{k+1})$ through the following QP:
\begin{equation} 
    \begin{aligned}
        \mathbf{k}_k(\mathbf{p}, t) = \underset{\mathbf{z} \in \mathbb{R}^z}{\arg\min}\,\,
        &\tfrac{1}{2}\|\mathbf{z} - \mathbf{k}_\mathrm{d}(\mathbf{p}, t)\|^2\\
        \mathrm{s.t.}\,\, 
        &\Dot{h}_k(\mathbf{p}, t, \mathbf{z}) \geq -\alpha(h_k(\mathbf{p}, t)),
    \end{aligned}
    \label{Eq:SolutionSafetyFilter1}
\end{equation}
where $\mathbf{k}_\mathrm{d}: \mathbb{R}^p\times \mathbb{R}_{\geq0} \to \mathbb{R}^z$ denotes the 
nominal controller for the agent. If the position of the agent at the instant $t_k$, denoted as 
$\mathbf{p}_k$, satisfies 
$\mathbf{p}_k \in \mathrm{int}(\cap_{i \in \mathcal{I}}\mathcal{C}_{i, k}(t_k))$, then we can 
select the parameter $\beta_k$ so that $h_k(\mathbf{p}_k, t_k) \geq 0$, thereby ensuring safety 
over the interval $[t_k, t_{k+1})$. Specifically, for a desired tolerance 
$\varepsilon_\beta \in \mathbb{R}_{>0}$ for the approximation in \eqref{Eq:OverallCBF}, $\beta_k$ 
can be chosen as\vspace{-0.01mm}
\begin{equation}
    \beta_k = \max\left\{\frac{\ln(M+1)+1}{\varepsilon_\beta}, 
    \frac{\ln(M+1)+1}{\min_{i \in \mathcal{\mathcal{I}}} h_{i, k}(\mathbf{p}_k, t_k)}\right\}.
    \label{Eq:SmoothingParameter}
\end{equation}
By \eqref{Eq:SafetyDetail}, this choice guarantees both the desired tolerance and strict safety at 
$t_k$, i.e., $h_k(\mathbf{p}_k, t_k) > 0$. Global safety for all $t \in \mathbb{R}_{\geq 0}$ is 
then ensured if, at every sampling time $t_k$, it holds that 
$\mathbf{p}_k \in \mathrm{int}(\cap_{i \in \mathcal{I}}\mathcal{C}_{i, k}(t_k))$. This condition 
will depend on the obstacle estimation algorithm, and we discuss how to formally guarantee it in 
Remark \ref{Rm:GlobalSafety}, at the end of Section VI.

For agents with second-order dynamics as in \eqref{Eq:AgentDynamics2}, the state $\mathbf{z}$ 
cannot be directly controlled. Thus, we must \textit{backstep} through $\mathbf{z}$ to design a 
CBF $h_{k, 1}: \mathbb{R}^p\times\mathbb{R}^z \times [t_k, t_{k+1}) \to \mathbb{R}$ for 
\eqref{Eq:AgentDynamics2}. The next subsection elaborates on this additional design step, and 
Algorithm \ref{Alg:Solution} summarizes the overall navigation scheme.


\subsection{CBF Backstepping} \label{Sec:Solution1A}

This subsection describes the backstepping design for agents with second-order dynamics. To this 
end, we begin by extending the backstepping result from \cite{taylor2022safe} to time-dependent 
CBFs.

\begin{proposition}[CBF Backstepping]
    Let $h: \mathbb{R}^p\times[t_0, t_\mathrm{f}) \to \mathbb{R}$ be a CBF for 
    \eqref{Eq:AgentDynamics1} on a set-valued flow 
    $\mathcal{D}: [t_0, t_\mathrm{f}) \rightrightarrows \mathbb{R}^p$, and let 
    $\mathbf{k}: \mathbb{R}^p\times [t_0, t_\mathrm{f}) \to \mathbb{R}^z$ denote a continuously 
    differentiable feedback controller such that, for all 
    $(\mathbf{p}, t) \in \mathcal{G}(\mathcal{D})$,
    \begin{equation}
        \Dot{h}(\mathbf{p}, t, \mathbf{k}(\mathbf{p}, t)) > -\alpha(h(\mathbf{p}, t)),
    \end{equation}
    where $\alpha$ is an extended class-$\mathcal{K}_\infty$ function. Then, the function 
    $h_1: \mathbb{R}^p\times\mathbb{R}^z\times[t_0, t_\mathrm{f}) \to \mathbb{R}$, defined as
    \begin{equation}
        h_1(\mathbf{p}, \mathbf{z}, t) = h(\mathbf{p}, t) 
        - \frac{1}{2\sigma}\|\mathbf{z} - \mathbf{k}(\mathbf{p}, t)\|^2
    \end{equation}
    for all 
    $(\mathbf{p}, \mathbf{z}, t) \in \mathbb{R}^p\times\mathbb{R}^z\times[t_0, t_\mathrm{f})$, 
    with $\sigma \in \mathbb{R}_{>0}$, is a CBF for the system \eqref{Eq:AgentDynamics2}, for an 
    associated extended class-$\mathcal{K}_\infty$ function $\alpha_1$ such that 
    $\alpha_1(s) \geq \alpha(s)$ for all $s \in \mathbb{R}$.
\end{proposition}

\newpage

\begin{proof}
    Let $\mathcal{C}_1: [t_0, t_\mathrm{f}) \rightrightarrows \mathbb{R}^p\times\mathbb{R}^z$ 
    denote the set-valued flow associated with $h_1$, defined for all $t \in [t_0, t_\mathrm{f})$ 
    as
    \begin{equation}
        \mathcal{C}_1(t) = \{(\mathbf{p}, \mathbf{z}) \in \mathbb{R}^p\times\mathbb{R}^z: 
        h_1(\mathbf{p}, \mathbf{z}, t) \geq 0\}.
    \end{equation}
    To show that $h_1$ is a CBF for the system \eqref{Eq:AgentDynamics2}, we have to verify that 
    it satisfies the following properties: (i) $h_1$ is continuously differentiable, (ii) 
    $\frac{\partial}{\partial(\mathbf{p}, \mathbf{z})} h_1(\mathbf{p}, \mathbf{z}, t) 
    \neq \mathbf{0}$ when 
    $h_1(\mathbf{p}, \mathbf{z}, t) = 0$, and (iii) there exists a set-valued flow 
    $\mathcal{D}_1: [t_0, t_\mathrm{f}) \rightrightarrows \mathbb{R}^p\times\mathbb{R}^z$ with 
    $\mathcal{C}_1(t) \subseteq \mathcal{D}_1(t)$ for all $t \in [t_0, t_\mathrm{f})$ and an 
    extended class-$\mathcal{K}_\infty$ function $\alpha_1$ such that, for all 
    $(\mathbf{p}, \mathbf{z}, t) \in \mathcal{G}(\mathcal{D}_1)$,
    \begin{equation}
        \sup_{\mathbf{u} \in \mathbb{R}^m}
        \Dot{h}_1(\mathbf{p}, \mathbf{z}, t, \mathbf{u}) 
        > -\alpha_1(h_1(\mathbf{p}, \mathbf{z}, t)).
        \label{Eq:DefinitionCBFBS}
    \end{equation}

    The properties (i) and (ii) are straightforward to verify since $\mathbf{k}$ is continuously 
    differentiable and $h$ is a CBF for the system \eqref{Eq:AgentDynamics1}. To show (iii), we 
    note that for $\mathbf{z} \neq \mathbf{k}(\mathbf{p}, t)$, the value of $\Dot{h}_1$ can be 
    made arbitrarily large through input, meaning that
    \begin{equation}
        \sup_{\mathbf{u} \in \mathbb{R}^m}
        \Dot{h}_1(\mathbf{p}, \mathbf{z}, t, \mathbf{u}) = \infty 
        > -\alpha(h_1(\mathbf{p}, \mathbf{z}, t)),
    \end{equation}
    and for $\mathbf{z} = \mathbf{k}(\mathbf{p}, t)$ and 
    $(\mathbf{p}, t) \in \mathcal{G}(\mathcal{D})$, we have
    \begin{equation}
        \scalebox{0.925}{$
        \Dot{h}_1(\mathbf{p}, \mathbf{k}(\mathbf{p}, t), t, \mathbf{u}) > -\alpha(h(\mathbf{p}, t))
        = -\alpha(h_1(\mathbf{p}, \mathbf{k}(\mathbf{p}, t), t)).$}
    \end{equation}
    Hence, the condition in \eqref{Eq:DefinitionCBFBS} holds for all 
    $(\mathbf{p}, t, \mathbf{z}) \in \mathcal{G}(\mathcal{D})\times\mathbb{R}^z$ and 
    an extended class-$\mathcal{K}_\infty$ function $\alpha_1$ with
    $\alpha_1(s) \geq \alpha(s)$ for all $s \in \mathbb{R}$, which concludes the proof. 
\end{proof}

\begin{remark}
    The CBF $h_1$ can be employed to render the set-valued flow $\mathcal{C}_1$ forward invariant 
    for the system \eqref{Eq:AgentDynamics2}. Therefore, to render the set-valued flow 
    $\mathcal{C}: [t_0, t_\mathrm{f}) \rightrightarrows \mathbb{R}^p$ associated with $h$ forward 
    invariant for the top-level subsystem of \eqref{Eq:AgentDynamics2}, the initial state 
    $(\mathbf{p}_0, \mathbf{z}_0)$ must satisfy $h_1(\mathbf{p}_0, \mathbf{z}_0, t_0) \geq 0$. If 
    the initial position lies in the interior of $\mathcal{C}(t_0)$, i.e.,
    $h(\mathbf{p}_0, t_0) > 0$, then this requirement can be satisfied by selecting
    \begin{equation}
        \sigma \geq \frac{1}{2h(\mathbf{p}_0, t_0)}
        \|\mathbf{z}_0 - \mathbf{k}(\mathbf{p}_0, t_0)\|^2.
    \end{equation}
\end{remark}

Building on the previous result, we can then design a CBF 
$h_{k, 1}: \mathbb{R}^p\times\mathbb{R}^z\times[t_k, t_{k+1}) \to \mathbb{R}$ for 
\eqref{Eq:AgentDynamics2} as
\begin{equation}
    h_{k, 1}(\mathbf{p}, \mathbf{z}, t) = h_k(\mathbf{p}, t) 
    - \frac{1}{2\sigma_k}\|\mathbf{z} - \mathbf{k}_k(\mathbf{p}, t)\|^2,
\end{equation}
where $\sigma_k \in \mathbb{R}_{>0}$. It is important to note, however, that in this formulation 
the top-level controller $\mathbf{k}_k$ must be continuously differentiable. Hence, $\mathbf{k}_k$ 
cannot be designed by means of a QP as in \eqref{Eq:SolutionSafetyFilter1} because QPs typically 
only guarantee local Lipschitz continuity. A few methods exist for constructing smooth safety 
filters (see, e.g., \cite{cohen2023characterizing}, \cite{ong2019universal}, \cite{li2024tunable}). 
Here, we adopt the one from \cite{ong2019universal}, which constructs a smooth controller based on 
Gaussian-weighted centroids. Specifically, consider the set
\begin{equation}
    K_{h_k}(\mathbf{p}, t) = \big\{\mathbf{z} \in \mathbb{R}^z: 
    \Dot{h}_k(\mathbf{p}, t, \mathbf{z}) \geq - \alpha(h_k(\mathbf{p}, t))\big\},
\end{equation}
which is a closed half-space since $\Dot{h}_k$ is affine in $\mathbf{z}$. Following 
\cite{ong2019universal}, we can design a top-level controller $\mathbf{k}_k$ as
\begin{equation} 
    \mathbf{k}_k(\mathbf{p}, t) = \bm{\mu}(K_{h_k}(\mathbf{p}, t)),
\end{equation}
where $\bm{\mu}$ is the Gaussian-weighted centroid function, given by
\begin{equation}
    \bm{\mu}(\mathcal{S}) = \frac{\int_\mathcal{S}\mathbf{z}
    \exp(-\|\mathbf{z}\|^2/(2\varsigma))d\mathbf{z}}
    {\int_\mathcal{S}\exp(-\|\mathbf{z}\|^2/(2\varsigma))d\mathbf{z}},
    \label{Eq:GWC}
\end{equation}
for $\mathcal{S} \subseteq \mathbb{R}^z$, $\varsigma \in \mathbb{R}_{>0}$. Since 
$K_{h_k}(\mathbf{p}, t)$ is a half-space, the numerator and denominator in \eqref{Eq:GWC} are the 
unnormalized mean

\newpage

\noindent and normalizing constant of a Gaussian truncated to a half-space, available in closed 
form (see, e.g., \cite{tallis1965plane}). The resulting controller is continuously differentiable 
if the agent dynamics, the CBF gradient, and $\alpha$ are continuously differentiable 
\cite{ong2019universal}.

When our agent has second-order dynamics as in \eqref{Eq:AgentDynamics2} and $h_{k, 1}$ is a valid 
CBF for it on 
$\mathcal{D}_{k, 1}: [t_k, t_{k+1}) \rightrightarrows \mathbb{R}^p\times\mathbb{R}^z$, we can then 
construct a safe controller $\mathbf{k}_{k, 1}: \mathcal{G}(\mathcal{D}_{k, 1}) \to \mathbb{R}^m$ 
for the sampling interval $[t_k, t_{k+1})$ through the following QP:
\begin{equation} 
    \begin{aligned}
        \mathbf{k}_{k, 1}(\mathbf{p}, \mathbf{z}, t) =\,\, 
        &\underset{\mathbf{u} \in \mathbb{R}^m}{\arg\min}\,\,
        \tfrac{1}{2}\|\mathbf{u} - \mathbf{k}_{\mathrm{d}, 1}(\mathbf{p}, \mathbf{z}, t)\|^2\\
        \mathrm{s.t.}\,\, 
        &\Dot{h}_{k, 1}(\mathbf{p}, \mathbf{z}, t, \mathbf{u}) 
        \geq -\alpha_1(h_{k, 1}(\mathbf{p}, \mathbf{z},t)),
    \end{aligned}
    \label{Eq:SolutionSafetyFilter2}
\end{equation}
where $\mathbf{k}_{\mathrm{d}, 1}: \mathbb{R}^p\times\mathbb{R}^z\times\mathbb{R}_{\geq0} 
\to \mathbb{R}^m$ is the nominal controller for the agent. If the position of the agent at time 
$t_k$ satisfies $\mathbf{p}_k \in \mathrm{int}(\cap_{i \in \mathcal{I}}\mathcal{C}_{i, k}(t_k))$, 
then \eqref{Eq:SmoothingParameter} ensures that $h_k(\mathbf{p}_k, t_k) > 0$. As a result, we can 
then select $\sigma_k$ so that $h_{k, 1}(\mathbf{p}_k, \mathbf{z}_k, t_k) \geq 0$, ensuring safety 
over the interval $[t_k, t_{k+1})$. For instance, given a nominal 
$\bar{\sigma} \in \mathbb{R}_{>0}$, $\sigma_k$ can be selected as
\begin{equation}
    \sigma_k = \max\left\{\bar{\sigma}, \tfrac{1}{2h_k(\mathbf{p}_k, t_k)}
    \|\mathbf{z}_k - \mathbf{k}_k(\mathbf{p}_k, t_k)\|^2\right\}.
\end{equation}
Global safety for all times $t \in \mathbb{R}_{\geq 0}$ is then achieved if, at 
every sampling time $t_k$, it always holds that 
$\mathbf{p}_k \in \mathrm{int}(\cap_{i \in \mathcal{I}}\mathcal{C}_{i, k}(t_k))$, as in the 
first-order case (see Remark \ref{Rm:GlobalSafety}, Section \ref{Sec:Solution3}). Here, 
$\bar\sigma$ trades off conservatism against control effort. A smaller $\bar\sigma$ 
weighs the quadratic error term in $h_{k,1}$ more heavily, yielding closer tracking of 
$\mathbf{k}_k(\mathbf{p},t)$ and more conservative behavior. In contrast, a larger $\bar\sigma$ 
relaxes this weighting for less conservative behavior, at the cost of more aggressive control near 
obstacles. For additional details on CBF backstepping, see \cite{taylor2022safe}, 
\cite{cohen2024constructive}.

\begin{remark}[Input Bounds]
    While not the focus of this paper, we note that control inputs are bounded in practice, 
    requiring an additional decision layer. One plausible way to handle this is via the 
    optimal-decay QP of \cite{zeng2021safety}, which can be applied to 
    \eqref{Eq:SolutionSafetyFilter1} and \eqref{Eq:SolutionSafetyFilter2}. For instance, for 
    first-order agents and an input constraint set $\mathcal{Z} \subset \mathbb{R}^z$, the QP in 
    \eqref{Eq:SolutionSafetyFilter1} becomes
    \begin{equation} 
        \begin{aligned}
            \mathbf{k}_k(\mathbf{p}, t) =\,\, 
            \underset{\mathbf{z} \in \mathcal{Z}, \,\omega \in \mathbb{R}}{\arg\min}\,\,
            &\tfrac{1}{2}\|\mathbf{z} - \mathbf{k}_\mathrm{d}(\mathbf{p}, t)\|^2 
            + \rho(\omega - 1)^2\\
            \mathrm{s.t.}\,\, 
            &\Dot{h}_k(\mathbf{p}, t, \mathbf{z}) \geq -\omega\alpha(h_k(\mathbf{p}, t)),
        \end{aligned}
    \end{equation}
    with $\rho \in \mathbb{R}_{>0}$, and analogously for second-order agents, where $\sigma_k$ can 
    also be included as a decision variable. This guarantees pointwise feasibility in the interior 
    of the safe set, though the analysis of recursive feasibility is left for future work. Backup 
    CBFs offer another plausible direction \cite{van2026generalizations, zhang2026safety}.
\end{remark}

\vspace{-1mm}


\section{Guaranteed Obstacle-State Estimation using Constrained Convex Generators} 
\label{Sec:Solution2}

This section details the first main component of our solution, which is the estimation algorithm 
used for obtaining guaranteed estimates of the obstacle states. For notational simplicity, in this 
section we omit the obstacle index $i$ from 
\eqref{Eq:ObstacleDynamics}-\eqref{Eq:MeasurementModel} and consider a generic linear model of the 
form
\begin{equation}
    \Dot{\mathbf{x}} = \mathbf{F}\mathbf{x} + \mathbf{w},
    \label{Eq:ObstacleDynamicsGeneric}
\end{equation}
with $\mathbf{x}, \mathbf{w} \in \mathbb{R}^n$, $\mathbf{F} \in \mathbb{R}^{n\times n}$, and the 
measurement model is
\begin{equation}
    \mathbf{y}_k = \mathbf{C}\mathbf{x}_k + \mathbf{v}_k,
    \label{Eq:MeasurementModelGeneric}
\end{equation}
with $\mathbf{C} \in \mathbb{R}^{y\times n}$ and $\mathbf{v}_k \in \mathbb{R}^{y}$. The subproblem 
addressed in this section can then be summarized as follows.

\newpage

\begin{subproblem}[Guaranteed Obstacle-State Estimation]
    Consider the linear system \eqref{Eq:ObstacleDynamicsGeneric} and the measurement model 
    \eqref{Eq:MeasurementModelGeneric}. Then, design an estimation algorithm that, at the sampling 
    time $t_k$, computes a state estimate $\hat{\mathcal{X}}_k$ such that 
    $\mathbf{x}_k \in \hat{\mathcal{X}}_k$, given the collected measurements and the compact 
    convex sets $\mathcal{X}_0$, $\mathcal{W}$, and $\mathcal{V}$, such that 
    $\mathbf{x}_0 \in \mathcal{X}_0$, $\mathbf{w} \in \mathcal{W}$, and 
    $\mathbf{v}_k \in \mathcal{V}$.
\end{subproblem}

\subsection{Constrained Convex Generators}

We first review the CCG class of sets, which constitutes the state-of-the-art framework for 
guaranteed state estimation.

\begin{definition}[CCG \cite{silvestre2021constrained}]
    A set $\mathcal{Z} \subset \mathbb{R}^n$ is a CCG if there exists a tuple 
    $(\mathbf{G}, \mathbf{c}, \mathbf{A}, \mathbf{b}) \in 
    \mathbb{R}^{n\times\xi}\times\mathbb{R}^n\times\mathbb{R}^{c\times\xi}\times\mathbb{R}^{c}$ 
    and a set $\mathfrak{G} = \mathcal{G}_1\times\mathcal{G}_2\times\dots\times\mathcal{G}_G 
    \subset \mathbb{R}^\xi$ such that
    \begin{equation}
        \mathcal{Z} = \{\mathbf{G}\bm{\xi} + \mathbf{c}: 
        \mathbf{A}\bm{\xi} = \mathbf{b}, \bm{\xi} \in \mathfrak{G}\},
        \label{Eq:DefinitionCCG}
    \end{equation}
    where, for each $j \in \{1, \dots, G\}$, the generator set $\mathcal{G}_j$ is defined as the 
    zero-sublevel set of a convex function $g_j: \mathbb{R}^{\xi_j} \to \mathbb{R}$:
    \begin{equation}
        \mathcal{G}_j = \left\{\bm{\xi}_j \in \mathbb{R}^{\xi_j}: g_j(\bm{\xi}_j) \leq 0\right\}.
        \label{Eq:CCGGeneratorSet}
    \end{equation}
\end{definition}

For a CCG $\mathcal{Z}$ defined as in \eqref{Eq:DefinitionCCG}, we introduce the shorthand 
notation $\mathcal{Z} = (\mathbf{G}, \mathbf{c}, \mathbf{A}, \mathbf{b}, \mathfrak{G}) \subset 
\mathbb{R}^n$. When $\mathcal{Z}$ is not subject to any equality constraint, we write 
$\mathcal{Z} = (\mathbf{G}, \mathbf{c}, [\,], [\,], \mathfrak{G})$. In practice, a CCG can be 
stored in a computer as a data structure with fields $\mathbf{G}$, $\mathbf{c}$, $\mathbf{A}$, and 
$\mathbf{b}$, a field holding the list of generator functions $\{g_1,\dots,g_G\}$, and a vector 
storing the dimension of each generator variable $\bm{\xi}_j$, with an implementation available in 
ReachTool (\url{https://github.com/danielmsilvestre/ReachTool}). The next result recalls the 
set-operation identities for CCGs.

\setlength{\arraycolsep}{1pt}
\begin{proposition}[Set Operations with CCGs \cite{silvestre2021constrained}]
    Consider a matrix $\mathbf{R} \in \mathbb{R}^{y\times n}$, a vector 
    $\mathbf{t} \in \mathbb{R}^y$, and three CCGs:
    \begin{itemize}
        \item $\mathcal{Z} = 
        (\mathbf{G}_z, \mathbf{c}_z, \mathbf{A}_z, \mathbf{b}_z, \mathfrak{G}_z)
        \subset \mathbb{R}^n$;
        \item $\mathcal{W} = 
        (\mathbf{G}_w, \mathbf{c}_w, \mathbf{A}_w, \mathbf{b}_w, \mathfrak{G}_w)
        \subset \mathbb{R}^n$;
        \item $\mathcal{V} = 
        (\mathbf{G}_v, \mathbf{c}_v, \mathbf{A}_v, \mathbf{b}_v, \mathfrak{G}_v)
        \subset \mathbb{R}^y$.
    \end{itemize}
    Then, the following identities hold:
    \begin{equation}
        \begin{aligned}
            \mathbf{R}\mathcal{Z} + \mathbf{t} &= (\mathbf{R}\mathbf{G}_z, 
            \mathbf{R}\mathbf{c}_z + \mathbf{t}, \mathbf{A}_z, \mathbf{b}_z, \mathfrak{G}_z),\\
            \mathcal{Z} \oplus \mathcal{W} &= 
            \scalebox{0.85}{$\left(
            \begin{bmatrix}
                \mathbf{G}_z & \mathbf{G}_w
            \end{bmatrix}\hspace{-1mm},
            \mathbf{c}_z + \mathbf{c}_w, 
            \begin{bmatrix}
                \mathbf{A}_z & \mathbf{0}\\
                \mathbf{0} & \mathbf{A}_w
            \end{bmatrix}\hspace{-1mm},
            \begin{bmatrix}
                \mathbf{b}_z\\
                \mathbf{b}_w
            \end{bmatrix}\hspace{-1mm},
            \mathfrak{G}_z\times\mathfrak{G}_w\right)$},\\
            \mathcal{Z} \cap_\mathbf{R} \mathcal{V} &= 
            \scalebox{0.79}{$\left(
            \begin{bmatrix}
                \mathbf{G}_z & \mathbf{0}
            \end{bmatrix}\hspace{-1mm},
            \mathbf{c}_z, 
            \begin{bmatrix}
                \mathbf{A}_z & \mathbf{0}\\
                \mathbf{0} & \mathbf{A}_v\\
                \mathbf{R}\mathbf{G}_z & -\mathbf{G}_v
            \end{bmatrix}\hspace{-1mm},
            \begin{bmatrix}
                \mathbf{b}_z\\
                \mathbf{b}_v\\
                \mathbf{c}_v - \mathbf{R}\mathbf{c}_z
            \end{bmatrix}\hspace{-1mm},
            \mathfrak{G}_z\times\mathfrak{G}_v\right)$}.
        \end{aligned}
        \label{Eq:SetOperationsCCGs}
    \end{equation}
\end{proposition}
\setlength{\arraycolsep}{2pt}
\vspace{1mm}

From the identities in \eqref{Eq:SetOperationsCCGs} and the fact that CCGs are convex by 
construction, it follows that CCGs are well suited for state estimation in linear systems. 
Moreover, CCGs constitute a very general class of sets, which significantly reduces the need for 
approximations. Particularly, they generalize many commonly used set classes, such as intervals, 
ellipsoids, CZs or polytopes, convex cones, or ellipsotopes (see the example in 
Fig.~\ref{Fig:CCGIllustration}).

Consequently, for the linear system \eqref{Eq:ObstacleDynamicsGeneric} and the observation model 
\eqref{Eq:MeasurementModelGeneric}, if the sets $\mathcal{X}_0$, $\mathcal{W}$, and $\mathcal{V}$ 
are CCGs, the guaranteed state estimation problem can be solved recursively by applying the 
identities in \eqref{Eq:SetOperationsCCGs}. Specifically, at each time $t_k$, the tightest estimate 
consistent with the available information is given by
\begin{equation}
    \hat{\mathcal{X}}_k = \left(\mathbf{\Phi}(T_s)\hat{\mathcal{X}}_{k-1} 
    \oplus \mathbf{\Gamma}(T_s)\tilde{\mathcal{W}}\right) 
    \cap_\mathbf{C} (\mathbf{y}_{k} - \mathcal{V}),
    \label{Eq:OptimalStateEstimate}
\end{equation}
with $\mathbf{\Phi}$, $\mathbf{\Gamma}$, $\tilde{\mathcal{W}}$ defined as in Section 
\ref{Sec:Solution1}, and an initial estimate
\begin{equation}
    \hat{\mathcal{X}}_0 = \mathcal{X}_0 \cap_\mathbf{C} (\mathbf{y}_{0} - \mathcal{V}).
\end{equation}

\newpage

\begin{figure}[H]
    \centering
    \includegraphics[width=0.31\linewidth]{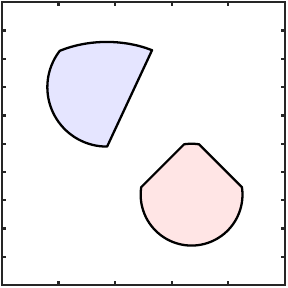}
    \caption{Two examples of CCGs with elipsoidal and polytopic components.}
    \label{Fig:CCGIllustration}
\end{figure}

\noindent The recursive formula in \eqref{Eq:OptimalStateEstimate} produces a new state estimate 
by first propagating the previous estimate over one sampling period and then updating the 
propagated set by intersecting it with the set of states consistent with the current measurement.

However, this approach has a major drawback. Since \eqref{Eq:OptimalStateEstimate} involves a 
Minkowski sum and generalized intersection, the number of generator sets increases over 
time, leading to a substantial computational load after a certain number of iterations. To address 
this issue, the next subsection introduces a more efficient approach that preserves a 
fixed-length representation.

\subsection{Explicit Finite-Horizon Estimator}

To keep a constant computational load over time, we adopt a finite-horizon estimation strategy that 
generalizes the approach recently introduced in \cite{rego2024explicit}. The main point is to fix 
a horizon length $N \in \mathbb{N}$ and consider an auxiliary conservative estimator which, at 
each sampling time $t_k \geq NT_s$, provides a CCG state estimate $\bar{\mathcal{X}}_{k-N}$ 
corresponding to the earlier instant $t_{k-N}$. The current estimate $\hat{\mathcal{X}}_k$ is then 
obtained by improving the estimate 
$\bar{\mathcal{X}}_{k-N} \cap_\mathbf{C} (\mathbf{y}_{k-N} - \mathcal{V})$ through $N$ iterations 
of \eqref{Eq:OptimalStateEstimate}.

However, note that a straightforward implementation of this approach would still require executing 
$N$ recursive iterations at each sampling instant $t_k \geq NT_s$, which would be computationally 
expensive for large $N$. The next result addresses this issue by showing that the CCG parameters 
of the current state estimate $\hat{\mathcal{X}}_k$ can be explicitly determined as a function of 
the parameters of the conservative estimate $\bar{\mathcal{X}}_{k-N}$.

\begin{theorem}[Explicit Computation of CCG Estimates]
    Consider the linear system \eqref{Eq:ObstacleDynamicsGeneric}, the measurement model 
    \eqref{Eq:MeasurementModelGeneric}, and let the sets $\tilde{\mathcal{W}}$ and $\mathcal{V}$ 
    be CCGs, defined as
    \begin{equation}
        \begin{aligned}
            \tilde{\mathcal{W}} &= (\mathbf{G}_{\tilde{w}}, \mathbf{c}_{\tilde{w}},
            [\,\,], [\,\,], \mathfrak{G}_{\tilde{w}}),\\
            \mathcal{V} &= (\mathbf{G}_v, \mathbf{c}_v, [\,\,], [\,\,], \mathfrak{G}_v).
        \end{aligned}
        \label{Eq:FiniteHorizonSets}
    \end{equation}
    Additionally, consider a fixed horizon $N \in \mathbb{N}$, and let
    \begin{equation}
        \bar{\mathcal{X}}_{k-N} = (\mathbf{G}_{\bar{x}, k-N}, \mathbf{c}_{\bar{x}, k-N}, 
        [\,\,], [\,\,], \mathfrak{G}_{\bar{x}, k-N})
        \label{Eq:ConservativeEstimate}
    \end{equation}
    be a state estimate corresponding to the sampling instant $t_{k-N}$. Then, the state estimate 
    at the sampling instant $t_k$, attained by improving 
    $\bar{\mathcal{X}}_{k-N} \cap_\mathbf{C} (\mathbf{y}_{k-N} - \mathcal{V})$ through $N$ 
    iterations of \eqref{Eq:OptimalStateEstimate}, is a CCG of the form
    \begin{equation}
        \hat{\mathcal{X}}_k = (\mathbf{G}_{\hat{x}, k}, \mathbf{c}_{\hat{x}, k}, 
        \mathbf{A}_{\hat{x}, k}, \mathbf{b}_{\hat{x}, k}, \mathfrak{G}_{\hat{x}, k}),
        \label{Eq:EstimatorEstimate}
    \end{equation}
    with parameters explicitly given by
    \begin{equation}
        \begin{aligned}
            \mathbf{G}_{\hat{x}, k} &=
            \begin{bmatrix}
                \mathbf{R}_{1, N}\mathbf{G}_{\bar{x}, k-N} & \mathbf{R}_{2, N}
            \end{bmatrix},\\
            \mathbf{c}_{\hat{x}, k} &= 
            \mathbf{R}_{1, N}\mathbf{c}_{\bar{x}, k-N} + \mathbf{t}_{1, N},\\
            \mathbf{A}_{\hat{x}, k} &=
            \begin{bmatrix}
                \mathbf{R}_{3, N}\mathbf{G}_{\bar{x}, k-N} & \mathbf{R}_{4, N}
            \end{bmatrix},\\
            \mathbf{b}_{\hat{x}, k} &= 
            \mathbf{R}_{5, N}\mathbf{c}_{\bar{x}, k-N} + \mathbf{t}_{2, N} + \mathbf{y}_{k-N:k},\\
            \mathfrak{G}_{\hat{x}, k} &= 
            \mathfrak{G}_{\bar{x}, k-N}\times\mathfrak{C}_N,
        \end{aligned}
        \label{Eq:EstimatorSetParameters}
    \end{equation}
    where $\mathbf{y}_{k-N:k} = (\mathbf{y}_{k-N}, \dots, \mathbf{y}_k)$, and 
    $\mathbf{R}_{1, N}, \dots, \mathbf{R}_{5, N}$, $\mathbf{t}_{1, N}$, 
    $\mathbf{t}_{2, N}$, and $\mathfrak{C}_{N}$ are obtained via $N$ iterations of the recursions
    \begin{equation}
        \begin{aligned}
            \mathbf{R}_{1, l+1} &= \mathbf{\Phi}(T_s)\mathbf{R}_{1, l},\\
            \mathbf{R}_{2, l+1} &=
            \begin{bmatrix}
                \mathbf{\Phi}(T_s)\mathbf{R}_{2, l} & 
                \mathbf{\Gamma}(T_s)\mathbf{G}_{\tilde{w}} & \mathbf{0}
            \end{bmatrix},\\
            \mathbf{R}_{3, l+1} &=
            \begin{bmatrix}
                \mathbf{R}_{3, l}\\
                \mathbf{C}\mathbf{R}_{1, l+1}
            \end{bmatrix},\\
            \mathbf{R}_{4, l+1} &=
            \begin{bmatrix}
                \mathbf{R}_{4, l} & \mathbf{0} & \mathbf{0}\\
                \mathbf{C}\mathbf{\Phi}(T_s)\mathbf{R}_{2, l} & 
                \mathbf{C}\mathbf{\Gamma}(T_s)\mathbf{G}_{\tilde{w}} & \mathbf{G}_v
            \end{bmatrix},\\
            \mathbf{R}_{5, l+1} &=
            \begin{bmatrix}
                \mathbf{R}_{5, l}\\
                -\mathbf{C}\mathbf{R}_{1, l+1}
            \end{bmatrix},\\
            \mathbf{t}_{1, l+1} &= \mathbf{\Phi}(T_s)\mathbf{t}_{1, l} 
            + \mathbf{\Gamma}(T_s)\mathbf{c}_{\tilde{w}},\\
            \mathbf{t}_{2, l+1} &=
            \begin{bmatrix}
                \mathbf{t}_{2, l}\\
                -\mathbf{C}\mathbf{t}_{1, l+1}-\mathbf{c}_v
            \end{bmatrix},\\
            \mathfrak{C}_{l+1} &= 
            \mathfrak{C}_{l}\times\mathfrak{G}_{\tilde{w}}\times\mathfrak{G}_v,
        \end{aligned}
        \label{Eq:EstimatorParameters}
    \end{equation}
    initialized with $\mathbf{R}_{1, 0} = \mathbf{I}$, $\mathbf{R}_{2, 0} = \mathbf{0}$,  
    $\mathbf{R}_{3, 0} = \mathbf{C}$, $\mathbf{R}_{4, 0} = \mathbf{G}_v$, 
    $\mathbf{R}_{5, 0} = -\mathbf{C}$, $\mathbf{t}_{1, 0} = \mathbf{0}$, 
    $\mathbf{t}_{2, 0} = -\mathbf{c}_v$, $\mathfrak{C}_0 = \mathfrak{G}_v$.
\end{theorem}

\begin{proof}
    To prove this result, we have to show that the parameters of 
    $\bar{\mathcal{X}}_{k-N} \cap_\mathbf{C} (\mathbf{y}_{k-N} - \mathcal{V})$ can be expressed in 
    the form of \eqref{Eq:EstimatorSetParameters} and that performing an iteration of 
    \eqref{Eq:OptimalStateEstimate} produces a CCG with parameters of the same form. For this 
    purpose, let
    \begin{equation}
        \hat{\mathcal{X}}_{\bar{k}} = \bar{\mathcal{X}}_{\bar{k}} 
        \cap_\mathbf{C} (\mathbf{y}_{\bar{k}} - \mathcal{V}),
    \end{equation}
    with $\bar{k} = k-N$. By the intersection identity in \eqref{Eq:SetOperationsCCGs}, we obtain
    \begin{equation}
        \begin{aligned}
            \mathbf{G}_{\hat{x}, \bar{k}} &=
            \begin{bmatrix}
                \mathbf{G}_{\bar{x}, \bar{k}} & \mathbf{0}
            \end{bmatrix} =
            \begin{bmatrix}
                \mathbf{R}_{1, 0}\mathbf{G}_{\bar{x}, \bar{k}} & \mathbf{R}_{2, 0}
            \end{bmatrix},\\
            \mathbf{c}_{\hat{x}, \bar{k}} &= \mathbf{c}_{\bar{x}, \bar{k}} =
            \mathbf{R}_{1, 0}\mathbf{c}_{\bar{x}, \bar{k}} + \mathbf{t}_{1, 0},\\
            \mathbf{A}_{\hat{x}, \bar{k}} &=
            \begin{bmatrix}
                \mathbf{C}\mathbf{G}_{\bar{x}, \bar{k}} & \mathbf{G}_v
            \end{bmatrix} =
            \begin{bmatrix}
                \mathbf{R}_{3, 0}\mathbf{G}_{\bar{x}, \bar{k}} & \mathbf{R}_{4, 0}
            \end{bmatrix},\\
            \mathbf{b}_{\hat{x}, \bar{k}} &= -\mathbf{C}\mathbf{c}_{\bar{x}, \bar{k}} 
            - \mathbf{c}_v + \mathbf{y}_{\bar{k}} = \mathbf{R}_{5, 0}\mathbf{c}_{\bar{x}, \bar{k}} 
            + \mathbf{t}_{2, 0} + \mathbf{y}_{\bar{k}},\\
            \mathfrak{G}_{\hat{x}, \bar{k}} &= \mathfrak{G}_{\bar{x}, \bar{k}}\times\mathfrak{G}_v 
            = \mathfrak{G}_{\bar{x}, \bar{k}}\times\mathfrak{C}_0,
        \end{aligned}
    \end{equation}
    with the auxiliary variables $\mathbf{R}_{1, 0} = \mathbf{I}$, 
    $\mathbf{R}_{2, 0} = \mathbf{0}$, $\mathbf{R}_{3, 0} = \mathbf{C}$, 
    $\mathbf{R}_{4, 0} = \mathbf{G}_v$, $\mathbf{R}_{5, 0} = -\mathbf{C}$, 
    $\mathbf{t}_{1, 0} = \mathbf{0}$, $\mathbf{t}_{2, 0} = -\mathbf{c}_v$, and 
    $\mathfrak{C}_0 = \mathfrak{G}_v$. Applying now an iteration of 
    \eqref{Eq:OptimalStateEstimate} to $\hat{\mathcal{X}}_{\bar{k}}$ via 
    \eqref{Eq:SetOperationsCCGs} yields a state estimate $\hat{\mathcal{X}}_{\bar{k}+1}$ whose CCG 
    parameters are given by
    \begin{equation}
        \begin{aligned}
            \mathbf{G}_{\hat{x}, \bar{k}+1} &=
            \begin{bmatrix}
                \mathbf{\Phi}(T_s)\mathbf{G}_{\hat{x}, \bar{k}} 
                & \mathbf{\Gamma}(T_s)\mathbf{G}_{\tilde{w}} & \mathbf{0}
            \end{bmatrix}\\ &=
            \begin{bmatrix}
                \mathbf{\Phi}(T_s)\mathbf{R}_{1, 0}\mathbf{G}_{\bar{x}, \bar{k}} 
                & \mathbf{\Phi}(T_s)\mathbf{R}_{2, 0} & \mathbf{\Gamma}(T_s)\mathbf{G}_{\tilde{w}} 
                & \mathbf{0}
            \end{bmatrix}\\ &=
            \begin{bmatrix}
                \mathbf{R}_{1, 1}\mathbf{G}_{\bar{x}, \bar{k}} & \mathbf{R}_{2, 1}
            \end{bmatrix},\\
            \mathbf{c}_{\hat{x}, \bar{k}+1} &= 
            \mathbf{\Phi}(T_s)\mathbf{c}_{\hat{x}, \bar{k}} 
            + \mathbf{\Gamma}(T_s)\mathbf{c}_{\tilde{w}}\\
            &= \mathbf{\Phi}(T_s)\mathbf{R}_{1, 0}\mathbf{c}_{\bar{x}, \bar{k}} 
            + \mathbf{\Phi}(T_s)\mathbf{t}_{1, 0} 
            + \mathbf{\Gamma}(T_s)\mathbf{c}_{\tilde{w}}\\ 
            &= \mathbf{R}_{1, 1}\mathbf{c}_{\bar{x}, \bar{k}} + \mathbf{t}_{1, 1},\\
            \mathbf{A}_{\hat{x}, \bar{k}+1} &=
            \begin{bmatrix}
                \mathbf{A}_{\hat{x}, \bar{k}} & \mathbf{0} & \mathbf{0}\\
                \mathbf{C}\mathbf{\Phi}(T_s)\mathbf{G}_{\hat{x}, \bar{k}} 
                & \mathbf{C}\mathbf{\Gamma}(T_s)\mathbf{G}_{\tilde{w}} & \mathbf{G}_v
            \end{bmatrix}\\ &= \scalebox{0.96}{$
            \begin{bmatrix}
                \mathbf{R}_{3, 0}\mathbf{G}_{\bar{x}, \bar{k}} & \mathbf{R}_{4, 0} 
                & \mathbf{0} & \mathbf{0}\\
                \mathbf{C}\mathbf{R}_{1, 1}\mathbf{G}_{\bar{x}, \bar{k}} 
                & \mathbf{C}\mathbf{\Phi}(T_s)\mathbf{R}_{2, 0} 
                & \mathbf{C}\mathbf{\Gamma}(T_s)\mathbf{G}_{\tilde{w}} & \mathbf{G}_v
            \end{bmatrix}$}\\ &=
            \begin{bmatrix}
                \mathbf{R}_{3, 1}\mathbf{G}_{\bar{x}, \bar{k}} & \mathbf{R}_{4, 1}
            \end{bmatrix},\\
            \mathbf{b}_{\hat{x}, \bar{k}+1} &= 
            \begin{bmatrix}
                \mathbf{b}_{\hat{x}, \bar{k}}\\
                -\mathbf{C}\mathbf{c}_{\hat{x}, \bar{k}+1} 
                - \mathbf{c}_v + \mathbf{y}_{\bar{k}+1}
            \end{bmatrix}\\ &=
            \begin{bmatrix}
                \mathbf{R}_{5, 0}\mathbf{c}_{\bar{x}, \bar{k}} 
                + \mathbf{t}_{2, 0} + \mathbf{y}_{\bar{k}}\\
                -\mathbf{C}\mathbf{R}_{1, 1}\mathbf{c}_{\bar{x}, \bar{k}} 
                -\mathbf{C}\mathbf{t}_{1, 1} - \mathbf{c}_v + \mathbf{y}_{\bar{k}+1}
            \end{bmatrix}\\ &=
            \mathbf{R}_{5, 1}\mathbf{c}_{\bar{x}, \bar{k}} + \mathbf{t}_{2, 1} 
            + \mathbf{y}_{\bar{k}:\bar{k}+1},\\
            \mathfrak{G}_{\hat{x}, \bar{k}+1} &= 
            \mathfrak{G}_{\hat{x}, \bar{k}}\times\mathfrak{G}_{\tilde{w}}\times\mathfrak{G}_v
            = \mathfrak{G}_{\bar{x}, \bar{k}}\times\mathfrak{C}_0
            \times\mathfrak{G}_{\tilde{w}}\times\mathfrak{G}_v\\
            &= \mathfrak{G}_{\bar{x}, \bar{k}}\times\mathfrak{C}_1,
        \end{aligned}
    \end{equation}

    \newpage

    \begin{algorithm}[H]
        \caption{Precomputation of the Estimator Parameters}
        \label{Alg:EstimatorParameters}
        \begin{algorithmic}[1]
            \Require $\mathbf{\Phi}(T_s)$, $\mathbf{\Gamma}(T_s)$, $\mathbf{C}$, 
            $\mathbf{G}_{\tilde{w}}$, $\mathbf{c}_{\tilde{w}}$, $\mathfrak{G}_{\tilde{w}}$, 
            $\mathbf{G}_v$, $\mathbf{c}_v$, $\mathfrak{G}_v$, $N$
            \State $\mathbf{R}_{1, 0} \gets \mathbf{I}$, $\mathbf{R}_{2, 0} \gets \mathbf{0}$,  
            $\mathbf{R}_{3, 0} \gets \mathbf{C}$, $\mathbf{R}_{4, 0} \gets \mathbf{G}_v$,\\ 
            $\mathbf{R}_{5, 0} \gets -\mathbf{C}$, $\mathbf{t}_{1, 0} \gets \mathbf{0}$, 
            $\mathbf{t}_{2, 0} \gets -\mathbf{c}_v$, $\mathfrak{C}_0 \gets \mathfrak{G}_v$
            \For{$l \gets 0$ to $N-1$}
                \State compute $\mathbf{R}_{1, l+1}, \dots, \mathbf{R}_{5, l+1}$, 
                $\mathbf{t}_{1, l+1}$, $\mathbf{t}_{2, l+1}$, $\mathfrak{C}_{l+1}$ 
                \eqref{Eq:EstimatorParameters}
            \EndFor
            \State\Return $\mathbf{R}_{1, N}, \dots, \mathbf{R}_{5, N}$, $\mathbf{t}_{1, N}$, 
            $\mathbf{t}_{2, N}$, $\mathfrak{C}_N$
        \end{algorithmic}
    \end{algorithm}
    \vspace{-2mm}
    \begin{algorithm}[H]
        \caption{Explicit Finite-Horizon Estimator}
        \label{Alg:ExplicitEstimator}
        \begin{algorithmic}[1]
            \Require $\mathbf{R}_{1, N}, \dots, \mathbf{R}_{5, N}$, $\mathbf{t}_{1, N}$, 
            $\mathbf{t}_{2, N}$, $\mathfrak{C}_N$, $\mathbf{y}_{0:k}$, $k$, $N$
            \If{$k < N$}
                \State compute $\bar{\mathcal{X}}_k$ using conservative estimator
                \State $\hat{\mathcal{X}}_k \gets \bar{\mathcal{X}}_k$
            \Else
                \State compute $\bar{\mathcal{X}}_{k-N}$ using conservative estimator
                \State compute $\hat{\mathcal{X}}_k$ using 
                \eqref{Eq:EstimatorEstimate}-\eqref{Eq:EstimatorSetParameters}
            \EndIf
            \State\Return $\hat{\mathcal{X}}_k$
        \end{algorithmic}
    \end{algorithm}
    
    \noindent where the auxiliary variables $\mathbf{R}_{1, 1}, \dots, \mathbf{R}_{5, 1}$, 
    $\mathbf{t}_{1, 1}$, $\mathbf{t}_{2, 1}$, and $\mathfrak{C}_{1}$ are given by the recursive 
    formulas in \eqref{Eq:EstimatorParameters}. Therefore, since the CCG parameters of 
    $\hat{\mathcal{X}}_{\bar{k}}$ can be written as in \eqref{Eq:EstimatorSetParameters}, and one 
    iteration of \eqref{Eq:OptimalStateEstimate} produces a CCG with parameters of the same form, 
    the result follows by induction.
\end{proof}

\begin{remark}
    For simplicity, in \eqref{Eq:FiniteHorizonSets} and \eqref{Eq:ConservativeEstimate}, we only 
    consider CCGs with no equality constraints, but the result can be readily extended to CCGs with 
    equality constraints in \eqref{Eq:FiniteHorizonSets} and \eqref{Eq:ConservativeEstimate}. 
    Nevertheless, note that the equality constraints can always be eliminated by explicitly writing 
    their solutions.
\end{remark}

As shown in the previous result, the recursive formulas in \eqref{Eq:EstimatorParameters} only 
involve fixed parameters, meaning that the variables 
$\mathbf{R}_{1, N}, \dots, \mathbf{R}_{5, N}$, $\mathbf{t}_{1, N}$, $\mathbf{t}_{2, N}$, and 
$\mathfrak{C}_{N}$ can be precomputed for a given horizon $N$. Using these precomputed parameters, 
the state estimate $\hat{\mathcal{X}}_k$ at each sampling instant $t_k \geq NT_s$ can then be 
explicitly computed from the conservative estimate $\bar{\mathcal{X}}_{k-N}$ using the expressions 
in \eqref{Eq:EstimatorSetParameters}. During the initial offset phase, when $t_k < NT_s$, a 
plausible approach is to obtain the current estimate directly from the conservative estimator. 
Algorithms \ref{Alg:EstimatorParameters} and \ref{Alg:ExplicitEstimator} outline the proposed 
finite-horizon scheme.

Regarding the conservative estimator, a practical choice is to consider a Luenberger observer, as 
suggested in \cite{rego2024explicit}. However, alternative approaches are also viable, provided 
they produce CCG estimates with a fixed-length representation.

\begin{remark}[Extension to Nonlinear Dynamics]
    To extend this approach to nonlinear obstacle dynamics, one plausible route is to linearize the 
    dynamics and propagate the resulting error bounds, as done in \cite{rego2020guaranteed}, though 
    this tends to yield conservative bounds highly dependent on the linearization point. Currently, 
    a more active direction relies on Hybrid Zonotopes (HZs) and functional decomposition, where 
    the graph set of the nonlinear dynamics is overapproximated by a union of polytopes represented 
    in HZ format \cite{siefert2025reachability}. This would replace the propagation term within 
    \eqref{Eq:OptimalStateEstimate}, yielding an HZ-valued estimate. A further appealing property 
    in this direction is sharpness, whereby the convex relaxation of an HZ coincides with its 
    convex hull \cite{glunt2025sharp}.
\end{remark}
\vspace{-4mm}


\section{\texorpdfstring{Control Barrier Functions for\\Constrained Convex Generators}{}} 
\label{Sec:Solution3}

As detailed in Section \ref{Sec:Solution1}, given a guaranteed estimate $\hat{\mathcal{X}}_k$ of 
an obstacle state at the sampling instant $t_k$, the evolution of the estimated obstacle set 
during the sampling interval $[t_k, t_{k+1})$ can be described by the set-valued flow 
$\hat{\mathcal{O}}^+_k: [t_k, t_{k+1}) \rightrightarrows \mathbb{R}^p$, which is defined for all 
$t \in [t_k, t_{k+1})$ as
\begin{equation}
    \hat{\mathcal{O}}^+_k(t) = \mathbf{E}\left(\mathbf{\Phi}(t - t_k)\hat{\mathcal{X}}_k 
    \oplus \mathbf{\Gamma}(t - t_k)\tilde{\mathcal{W}}\right) \oplus \bar{\mathcal{O}}^+,
    \label{Eq:ObstacleSetEstimates}
\end{equation}
where the obstacle index $i$ has been omitted for brevity. Following the state estimation approach 
described in the previous section, the sets $\tilde{\mathcal{W}}$ and $\hat{\mathcal{X}}_k$ are 
CCGs, defined by \eqref{Eq:FiniteHorizonSets} and \eqref{Eq:EstimatorEstimate}. Thus, if the body 
set $\bar{\mathcal{O}}^+$ is also modeled as a CCG, then the estimated obstacle set 
$\hat{\mathcal{O}}^+_k(t)$ is a CCG and can be expressed in closed form using the set-operation 
identities from \eqref{Eq:SetOperationsCCGs}.

More specifically, let the body set be defined as
\begin{equation}
    \bar{\mathcal{O}}^+ = \left(\mathbf{G}_{\bar{o}}^+, \mathbf{c}_{\bar{o}}^+, 
    [\,\,], [\,\,], \mathfrak{G}_{\bar{o}}^+\right).
\end{equation}
Then, by applying the identities from \eqref{Eq:SetOperationsCCGs} to 
\eqref{Eq:ObstacleSetEstimates}, we conclude that the estimated obstacle set at each time
$t \in [t_k, t_{k+1})$ is a CCG of the form
\begin{equation}
    \hat{\mathcal{O}}^+_k(t) = \left(\mathbf{G}_{\hat{o}, k}^+(t), \mathbf{c}_{\hat{o}, k}^+(t), 
    \mathbf{A}_{\hat{o}, k}^+, \mathbf{b}_{\hat{o}, k}^+, \mathfrak{G}_{\hat{o}, k}^+\right),
\end{equation}
where the CCG parameters are given by
\begin{equation}
    \begin{aligned}
        \mathbf{G}_{\hat{o}, k}^+(t) &= 
        \begin{bmatrix}
            \mathbf{E}\mathbf{\Phi}(t - t_k)\mathbf{G}_{\hat{x}, k} 
            & \mathbf{E}\mathbf{\Gamma}(t - t_k)\mathbf{G}_{\tilde{w}}
            & \mathbf{G}_{\bar{o}}^+
        \end{bmatrix},\\
        \mathbf{c}_{\hat{o}, k}^+(t) &= \mathbf{E}\mathbf{\Phi}(t - t_k)\mathbf{c}_{\hat{x}, k} 
        + \mathbf{E}\mathbf{\Gamma}(t - t_k)\mathbf{c}_{\tilde{w}} 
        + \mathbf{c}_{\bar{o}}^+,\\
        \mathbf{A}_{\hat{o}, k}^+ &=
        \begin{bmatrix}
            \mathbf{A}_{\hat{x}, k} & \mathbf{0} & \mathbf{0}
        \end{bmatrix},\\
        \mathbf{b}_{\hat{o}, k}^+ &= \mathbf{b}_{\hat{x}, k},\\
        \mathfrak{G}_{\hat{o}, k}^+ &= 
        \mathfrak{G}_{\hat{x}, k}\times\mathfrak{G}_{\tilde{w}}\times\mathfrak{G}_{\bar{o}}^+,
    \end{aligned}
\end{equation}
from which it becomes clear that the functions $\mathbf{G}_{\hat{o}, k}^+$ and 
$\mathbf{c}_{\hat{o}, k}^+$ are (infinitely) continuously differentiable. Moreover, since 
$\tilde{\mathcal{W}}$, $\hat{\mathcal{X}}_k$, and $\bar{\mathcal{O}}^+$ are compact sets, and the 
body set has a nonempty interior, it follows that the estimated obstacle set is a compact set with 
nonempty interior.

With the set-valued flow $\hat{\mathcal{O}}^+_k$ in hand, the next step is then to construct a CBF 
that captures the safety objective of avoiding the respective obstacle. However, since the 
estimated obstacle set at each time instant is a CCG, we cannot directly convert 
$\hat{\mathcal{O}}^+_k$ into a CBF. Accordingly, the subproblem addressed in this section is 
stated as follows.

\begin{subproblem}[CBFs for CCGs]
    Let $\mathcal{O}: [t_0, t_\mathrm{f}) \rightrightarrows \mathbb{R}^p$ be a set-valued flow 
    defined for all $t \in [t_0, t_\mathrm{f})$ as
    \begin{equation}
        \begin{aligned}
            \mathcal{O}(t) &= \{\mathbf{G}(t)\bm{\xi} + \mathbf{c}(t): 
            \mathbf{A}\bm{\xi} = \mathbf{b}, \bm{\xi} \in \mathfrak{G}\}\\
            &= (\mathbf{G}(t), \mathbf{c}(t), \mathbf{A}, \mathbf{b}, \mathfrak{G}),
        \end{aligned}
        \label{Eq:Subproblem2}
    \end{equation}
    where $\mathcal{O}(t)$ is a compact set with nonempty interior, the functions 
    $\mathbf{G}: [t_0, t_\mathrm{f}) \to \mathbb{R}^{p\times\xi}$, 
    $\mathbf{c}: [t_0, t_\mathrm{f}) \to \mathbb{R}^p$ are continuously differentiable, 
    $\mathbf{A} \in \mathbb{R}^{c\times\xi}$, $\mathbf{b} \in \mathbb{R}^c$, and 
    $\mathfrak{G} = \mathcal{G}_1\times\mathcal{G}_2\times\dots\times\mathcal{G}_G$, where, for 
    each $j \in \mathcal{J} = \{1, \dots, G\}$, the generator set $\mathcal{G}_j$ is the 
    zero-sublevel set of a convex function $g_j: \mathbb{R}^{\xi_j} \to \mathbb{R}$. Then, convert 
    $\mathcal{O}$ into a CBF $h: \mathbb{R}^p\times[t_0, t_\mathrm{f}) \to \mathbb{R}$ for 
    \eqref{Eq:AgentDynamics1}, with an associated safe set-valued flow 
    $\mathcal{C}: [t_0, t_\mathrm{f}) \rightrightarrows \mathbb{R}^p$ such that 
    $\mathcal{C}(t) \subseteq \mathbb{R}^p\setminus\mathcal{O}(t)$ for all 
    $t \in [t_0, t_\mathrm{f})$.
\end{subproblem}

The next subsection details the proposed conversion process, using the generic notation in the 
previous statement for clarity.

\subsection{Conversion Procedure}

The first step of the conversion process is to eliminate the linear equality constraint from 
\eqref{Eq:Subproblem2}. Since $\mathcal{O}(t)$ has nonempty interior, this linear equation admits 
infinitely many solutions, which can be expressed in the parametric form
\begin{equation}
    \bm{\xi} = \mathbf{A}^\dagger\mathbf{b} + \mathrm{null}(\mathbf{A})\bm{\eta},\quad 
    \bm{\eta} \in \mathbb{R}^\eta.
\end{equation}
As no assumptions are made about the rank of $\mathbf{A}$, its pseudoinverse is obtained via 
Singular Value Decomposition (SVD):
\begin{equation}
    \mathbf{A} = \mathbf{U}_\mathbf{A}\mathbf{\Sigma}_\mathbf{A}\mathbf{V}_\mathbf{A}^\top
    \;\;\implies\;\;
    \mathbf{A}^\dagger = 
    \mathbf{V}_\mathbf{A}\mathbf{\Sigma}_\mathbf{A}^{-1}\mathbf{U}_\mathbf{A}^\top,
\end{equation}
where $\mathbf{U}_\mathbf{A}$ and $\mathbf{V}_\mathbf{A}$ have orthonormal columns spanning the 
column and row spaces of $\mathbf{A}$, respectively, and $\mathbf{\Sigma}_\mathbf{A}$ is diagonal, 
holding the singular values. Using this parameterization, $\mathcal{O}(t)$ can be equivalently 
expressed with no equality constraint as
\begin{equation}
    \mathcal{O}(t) = \left(\mathbf{G}(t)\mathrm{null}(\mathbf{A}), 
    \mathbf{c}(t) + \mathbf{G}(t)\mathbf{A}^\dagger\mathbf{b}, [\,\,], [\,\,], 
    \mathfrak{G}'\right).
\end{equation}
Here, the new generator set $\mathfrak{G}'$ is given by
\begin{equation}
    \mathfrak{G}' = \{\bm{\eta} \in \mathbb{R}^{\eta}: f_j(\bm{\eta}) \leq 0 
    \text{ for all } j \in \mathcal{J}\},
\end{equation}
where each function $f_j: \mathbb{R}^{\eta} \to \mathbb{R}$ is defined as
\begin{equation}
    f_j(\bm{\eta}) = g_j\!\left(\mathbf{S}_j\mathbf{A}^\dagger\mathbf{b} 
    + \mathbf{S}_j\mathrm{null}(\mathbf{A})\bm{\eta}\right)
\end{equation}
for all $\bm{\eta} \in \mathbb{R}^{\eta}$, where 
$\mathbf{S}_j \in \mathbb{R}^{\xi_j\times\xi}$ is an auxiliary matrix that selects the $j$th 
individual generator vector $\bm{\xi}_j$ from the overall generator vector 
$\bm{\xi} = (\bm{\xi}_1, \dots, \bm{\xi}_G)$. Moreover, since $\mathcal{O}(t)$ has nonempty 
interior, the matrix $\mathbf{G}(t)\mathrm{null}(\mathbf{A})$ has full row rank.

Observe now that the generator set $\mathfrak{G}'$ can be equivalently expressed using the maximum 
function as
\begin{equation}
    \mathfrak{G}' = \left\{\bm{\eta} \in \mathbb{R}^{\eta}: 
    \max_{j \in \mathcal{\mathcal{J}}}\,f_j(\bm{\eta}) \leq 0\right\}.
\end{equation}
However, since the maximum function is not differentiable, it cannot be directly used in the CBF 
design. To overcome this, we apply a smooth underapproximation of the maximum func- 

\noindent tion, adopting the LogSumExp approach once more from \cite{molnar2023composing}. 
Specifically, we define a set-valued flow 
$\tilde{\mathcal{O}}: [t_0, t_\mathrm{f}) \rightrightarrows \mathbb{R}^p$ as
\begin{equation}
    \tilde{\mathcal{O}}(t) = \left(\tilde{\mathbf{G}}(t), \tilde{\mathbf{c}}(t), [\,\,], [\,\,], 
    \tilde{\mathfrak{G}}\right)
\end{equation}
for all $t \in [t_0, t_\mathrm{f})$, with parameters given by
\begin{equation}
    \begin{aligned}
        \tilde{\mathbf{G}}(t) &= \mathbf{G}(t)\mathrm{null}(\mathbf{A}),\\
        \tilde{\mathbf{c}}(t) &= \mathbf{c}(t) + \mathbf{G}(t)\mathbf{A}^\dagger\mathbf{b},\\
        \tilde{\mathfrak{G}} &= \{\bm{\eta} \in \mathbb{R}^{\eta}: f(\bm{\eta}) \leq 0\},
    \end{aligned}
\end{equation}
where the function $f: \mathbb{R}^{\eta} \to \mathbb{R}$ is defined for all 
$\bm{\eta} \in \mathbb{R}^\eta$ as
\begin{equation}
    f(\bm{\eta}) = \frac{1}{\gamma}
    \ln\!\left(\sum_{j \in \mathcal{J}}\exp(\gamma f_j(\bm{\eta}))\right) - \frac{\ln(G+1)}{\gamma},
\end{equation}
with smoothing parameter $\gamma \in \mathbb{R}_{>0}$. This construction follows from 
\cite{molnar2023composing} and ensures that, for all $\bm{\eta} \in \mathbb{R}^\eta$, we have
\begin{equation}
     0 < \max_{j \in \mathcal{\mathcal{J}}} f_j(\bm{\eta}) - f(\bm{\eta}) 
     \leq \frac{\ln(G+1)}{\gamma},
\end{equation}
which implies $\mathcal{O}(t) \subset \mathrm{int}\big(\tilde{\mathcal{O}}(t)\big)$ and 
$\lim_{\gamma\to\infty} \tilde{\mathcal{O}}(t) = \mathcal{O}(t)$. In particular, $\gamma$ can be 
chosen to satisfy a desired conservatism tolerance $\varepsilon_\gamma \in \mathbb{R}_{>0}$ by 
setting $\gamma \geq \ln(G+1)/\varepsilon_\gamma$. Moreover, note that $f$ is convex, as each 
function $f_j$ is convex and $f$ is defined through convexity-preserving operations.

\newpage

Since the goal is to obtain a CBF for the agent dynamics in \eqref{Eq:AgentDynamics1}, we must now 
express the set $\tilde{\mathcal{O}}(t)$ in terms of the agent’s position. To this end, note that 
a point $\mathbf{p} \in \mathbb{R}^p$ belongs to $\tilde{\mathcal{O}}(t)$ if and only if there 
exists some $\bm{\eta} \in \mathbb{R}^\eta$ so that $f(\bm{\eta}) \leq 0$ and 
$\tilde{\mathbf{G}}(t)\bm{\eta} + \tilde{\mathbf{c}}(t) = \mathbf{p}$. Formally, this means that
\begin{equation}
    \scalebox{0.89}{$
    \tilde{\mathcal{O}}(t) = \left\{\mathbf{p} \in \mathbb{R}^p: 
    \exists \bm{\eta} \in \mathbb{R}^\eta: f(\bm{\eta}) \leq 0,
    \tilde{\mathbf{G}}(t)\bm{\eta} + \tilde{\mathbf{c}}(t) = \mathbf{p}\right\}.$}
\end{equation}
Consequently, this set can be equivalently written using a CBF candidate 
$h: \mathbb{R}^p\times[t_0, t_\mathrm{f}) \to \mathbb{R}$ as
\begin{equation}
     \tilde{\mathcal{O}}(t) = \{\mathbf{p} \in \mathbb{R}^p: h(\mathbf{p}, t) \leq 0\},
     \label{Eq:CBFCCGFlow}
\end{equation}
where $h$ is defined via the following optimization problem:
\begin{equation} 
    h(\mathbf{p}, t) = \min_{\bm{\eta} \in \mathbb{R}^\eta} f(\bm{\eta}) 
    \quad\mathrm{s.t.}\quad \tilde{\mathbf{G}}(t)\bm{\eta} + \tilde{\mathbf{c}}(t) = \mathbf{p},
    \label{Eq:CBFCCG}
\end{equation}
provided that \eqref{Eq:CBFCCG} is solvable for all 
$(\mathbf{p}, t) \in \mathbb{R}^p\times[t_0, t_\mathrm{f})$. The following result provides a mild 
sufficient condition for $h$ to qualify as a valid CBF for the agent dynamics in 
\eqref{Eq:AgentDynamics1}.

\begin{theorem} \label{Th:TheoremCBGCCG}
    Let $\tilde{\mathcal{O}}: [t_0, t_\mathrm{f}) \rightrightarrows \mathbb{R}^p$ be a set-valued 
    flow defined as in \eqref{Eq:CBFCCGFlow}, where $\tilde{\mathcal{O}}(t)$ has nonempty interior 
    for all $t \in [t_0, t_\mathrm{f})$, and 
    $h: \mathbb{R}^p\times[t_0, t_\mathrm{f}) \to \mathbb{R}$ is defined as in 
    \eqref{Eq:CBFCCG}, where the maps
    $\tilde{\mathbf{G}}: [t_0, t_\mathrm{f}) \to \mathbb{R}^{p\times\eta}$ and
    $\tilde{\mathbf{c}}: [t_0, t_\mathrm{f}) \to \mathbb{R}^p$ are continuously differentiable, and 
    $\tilde{\mathbf{G}}(t)$ has full row rank for all 
    $t \in [t_0, t_\mathrm{f})$. If $f: \mathbb{R}^{\eta} \to \mathbb{R}$ is twice continuously 
    differentiable and strictly convex, then $h$ is a CBF for the system \eqref{Eq:AgentDynamics1}.
\end{theorem}

\begin{proof}
    Let $\mathcal{C}: [t_0, t_\mathrm{f}) \rightrightarrows \mathbb{R}^p$ be the set-valued flow 
    associated with $h$, defined for all $t \in [t_0, t_\mathrm{f})$ as 
    $\mathcal{C}(t) = \mathbb{R}^p\setminus\mathrm{int}\big(\tilde{\mathcal{O}}(t)\big)$. To show 
    that $h$ is a CBF for \eqref{Eq:AgentDynamics1}, we have to verify that it satisfies the 
    following properties: (i) $h$ is continuously differentiable, (ii) 
    $\frac{\partial}{\partial\mathbf{p}}h(\mathbf{p}, t) \neq \mathbf{0}$ when 
    $h(\mathbf{p}, t) = 0$, and (iii) there exists an extended class-$\mathcal{K}_\infty$ function 
    $\alpha$ such that
    \begin{equation}
        \sup_{\mathbf{z} \in \mathbb{R}^z} \Dot{h}(\mathbf{p}, t) > -\alpha(h(\mathbf{p}, t))
        \label{Eq:CBFCCG3}
    \end{equation}
    for all $(\mathbf{p}, t) \in \mathcal{G}(\mathcal{C})$, with $\Dot{h}$ computed with respect to 
    \eqref{Eq:AgentDynamics1}.

    The proof of property (i) follows the same steps as Fiacco's sensitivity theorem 
    \cite{fiacco1983introduction}. Since the optimization problem \eqref{Eq:CBFCCG} is convex, 
    its solutions coincide with those of the respective KKT conditions. In this case, the KKT 
    system follows as
    \begin{equation}
        \bm{\psi}(\mathbf{p}, t, \bm{\eta}, \bm{\lambda}) :=
        \begin{bmatrix}
            \left(\dfrac{\partial f(\bm{\eta})}{\partial\bm{\eta}}\right)^\top
            + \tilde{\mathbf{G}}(t)^\top\bm{\lambda}\\[3mm]
            \tilde{\mathbf{G}}(t)\bm{\eta} + \tilde{\mathbf{c}}(t) - \mathbf{p}
        \end{bmatrix} = \mathbf{0},
        \label{Eq:CBFKKT}
    \end{equation}
    where $\bm{\lambda} \in \mathbb{R}^{p}$ is the Lagrange multiplier associated with the 
    equality constraint. From this setup, the partial derivatives of $\bm{\psi}$ with respect to 
    $(\mathbf{p}, t)$ and $(\bm{\eta}, \bm{\lambda})$ are given by
    \begin{equation}
        \begin{aligned}
            \frac{\partial \bm{\psi}(\mathbf{p}, t, \bm{\eta}, \bm{\lambda})}
            {\partial(\mathbf{p}, t)} &=
            \begin{bmatrix}
                \mathbf{0} & \Dot{\tilde{\mathbf{G}}}(t)^\top\bm{\lambda}\\
                -\mathbf{I} & 
                \Dot{\tilde{\mathbf{G}}}(t)\bm{\eta} + \Dot{\tilde{\mathbf{c}}}(t)
            \end{bmatrix},\\
            \frac{\partial \bm{\psi}(\mathbf{p}, t, \bm{\eta}, \bm{\lambda})}
            {\partial (\bm{\eta}, \bm{\lambda})} &=
            \begin{bmatrix}
                \dfrac{\partial^2 f(\bm{\eta})}{\partial\bm{\eta}^2}
                & \tilde{\mathbf{G}}(t)^\top\\[2.5mm]
                \tilde{\mathbf{G}}(t) & \mathbf{0}
            \end{bmatrix}.
        \end{aligned}
    \end{equation}
    Applying now the Schur complement to the latter yields
    \begin{equation}
        \scalebox{0.77}{$
        \begin{aligned}
            \det\!\left(\frac{\partial \bm{\psi}(\cdot)}
            {\partial (\bm{\eta}, \bm{\lambda})}\right) =
            -\det\!\left(\frac{\partial^2 f(\bm{\eta})}
            {\partial\bm{\eta}^2}\right)\det\!\left(\tilde{\mathbf{G}}(t)
            \left(\frac{\partial^2 f(\bm{\eta})}
            {\partial\bm{\eta}^2}\right)^{-1}\tilde{\mathbf{G}}(t)^\top\right).
        \end{aligned}$}
    \end{equation}
    Since $f$ is strictly convex by assumption, meaning the Hessian matrix \hfill of \hfill $f$ 
    \hfill is \hfill positive \hfill definite \hfill for \hfill all \hfill 
    $\bm{\eta} \in \mathbb{R}^\eta$, \hfill and \hfill $\tilde{\mathbf{G}}(t)$ \hfill has

    \newpage
    
    \noindent has 
    full row rank for all $t \in [t_0, t_\mathrm{f})$, the previous determinant is always nonzero. 
    Therefore, by the Implicit Function Theorem, the KKT system defines 
    $(\bm{\eta}, \bm{\lambda})$ as a function of $(\mathbf{p}, t)$, i.e.,
    \begin{equation}
        (\bm{\eta}, \bm{\lambda}) = \bm{\ell}(\mathbf{p}, t),
    \end{equation}
    where the implicit function $\bm{\ell}: \mathbb{R}^p\times[t_0, t_\mathrm{f}) \to 
    \mathbb{R}^\eta\times\mathbb{R}^p$ is continuously differentiable, with derivative given by
    \begin{equation}
        \frac{\partial \bm{\ell}(\mathbf{p}, t)}{\partial(\mathbf{p}, t)} 
        = - \left[\left(
        \frac{\partial \bm{\psi}(\cdot)}{\partial (\bm{\eta}, \bm{\lambda})}\right)^{-1} 
        \frac{\partial \bm{\psi}(\cdot)}{\partial (\mathbf{p}, t)}\right]_{
        (\bm{\eta}, \bm{\lambda}) = \bm{\ell}(\mathbf{p}, t)}.
        \label{Eq:Numerical1}
    \end{equation}
    Consequently, $h$ can be equivalently expressed as
    \begin{equation}
        h(\mathbf{p}, t) = f(\mathbf{E}_{\bm{\eta}}\bm{\ell}(\mathbf{p}, t)),
        \label{Eq:Numerical2}
    \end{equation}
    where $\mathbf{E}_{\bm{\eta}} \in \mathbb{R}^{\eta\times(\eta+p)}$ extracts $\bm{\eta}$ from 
    $(\bm{\eta}, \bm{\lambda})$, from which it becomes clear that $h$ is continuously 
    differentiable, with
    \begin{equation}
        \frac{\partial h(\mathbf{p}, t)}{\partial(\mathbf{p}, t)} = 
        \frac{\partial f(\bm{\eta})}{\partial\bm{\eta}}
        \bigg|_{\bm{\eta}=\mathbf{E}_{\bm{\eta}}\bm{\ell}(\mathbf{p}, t)}\mathbf{E}_{\bm{\eta}}
        \frac{\partial \bm{\ell}(\mathbf{p}, t)}{\partial(\mathbf{p}, t)}.
        \label{Eq:Numerical3}
    \end{equation}

    Finally, we observe that $h$ is convex in $\mathbf{p}$, as the epigraph of $h(\cdot, t)$ is 
    the projection of the set
    \begin{equation}
        \scalebox{0.88}{$
        \mathcal{S}(t) = \left\{(\mathbf{p}, s, \bm{\eta}) 
        \in \mathbb{R}^{p+1+\eta}: f(\bm{\eta}) \leq s,
        \tilde{\mathbf{G}}(t)\bm{\eta} + \tilde{\mathbf{c}}(t) = \mathbf{p}\right\}$}
    \end{equation}
    in the $(\mathbf{p}, s)$-space, which is a convex set since $\mathcal{S}(t)$ is convex. 
    Thus, as $h$ is convex in $\mathbf{p}$ and $\tilde{\mathcal{O}}(t)$ has nonempty interior, it 
    follows that $\frac{\partial}{\partial\mathbf{p}} h(\mathbf{p}, t) \neq \mathbf{0}$ when 
    $h(\mathbf{p}, t) \geq 0$. Consequently, for any extended class-$\mathcal{K}_\infty$ function 
    $\alpha$, the inequality in \eqref{Eq:CBFCCG3} holds for all 
    $(\mathbf{p}, t) \in \mathcal{G}(\mathcal{C})$, as $\Dot{h}$ can be made arbitrarily large 
    when $\frac{\partial}{\partial\mathbf{p}} h(\mathbf{p}, t) \neq \mathbf{0}$. Hence, properties 
    (ii) and (iii) are also satisfied, which completes the proof.
\end{proof}

The preceding result establishes a mild sufficient condition under which the function $h$, defined 
in \eqref{Eq:CBFCCG}, qualifies as a CBF for the agent dynamics in \eqref{Eq:AgentDynamics1}. 
Particularly, if each function $g_j$ is twice continuously differentiable and strictly convex, 
then $f$ inherits these properties and, by Theorem \ref{Th:TheoremCBGCCG}, $h$ constitutes a valid 
CBF for \eqref{Eq:AgentDynamics1}. To verify this claim, observe that
\begin{equation}
    \scalebox{0.915}{$
    \begin{aligned}
        \left(\frac{\partial f_j(\bm{\eta})}{\partial\bm{\eta}}\right)^\top &= 
        \mathrm{null}(\mathbf{A})^\top\mathbf{S}_j^\top\frac{\partial g_j(\bm{\xi}_j)}
        {\partial\bm{\xi}_j}\bigg|^\top_{\bm{\xi}_j=\mathbf{s}_j(\bm{\eta})},\\[1mm]
        \frac{\partial^2f_j(\bm{\eta})}{\partial\bm{\eta}^2} &= 
        \mathrm{null}(\mathbf{A})^\top\mathbf{S}_j^\top\frac{\partial^2g_j(\bm{\xi}_j)}
        {\partial\bm{\xi}_j^2}\bigg|_{\bm{\xi}_j=\mathbf{s}_j(\bm{\eta})}
        \mathbf{S}_j\mathrm{null}(\mathbf{A}),
    \end{aligned}$}
\end{equation}
where $\mathbf{s}_j(\bm{\eta}) = \mathbf{S}_j\mathbf{A}^\dagger\mathbf{b} 
+ \mathbf{S}_j\mathrm{null}(\mathbf{A})\bm{\eta}$, and also note that
\begin{equation}
    \frac{\partial f(\bm{\eta})}{\partial\bm{\eta}} = \sum_{j \in \mathcal{J}} 
    \pi_j(\bm{\eta})\frac{\partial f_j(\bm{\eta})}{\partial\bm{\eta}},
\end{equation}
where each weight is defined as
\begin{equation}
    \pi_j(\bm{\eta}) = \exp(\gamma(f_j(\bm{\eta}) - f(\bm{\eta})) - \ln(G+1)).
\end{equation}
Consequently, the Hessian matrix of $f$ follows as
\begin{equation}
    \scalebox{0.95}{$
    \begin{aligned}
        \frac{\partial^2f(\bm{\eta})}{\partial\bm{\eta}^2} &= 
        \mathrm{null}(\mathbf{A})^\top\,\underset{j \in \mathcal{J}}{\mathrm{diag}}\!
        \left(\pi_j(\bm{\eta})\frac{\partial^2g_j(\bm{\xi}_j)}{\partial\bm{\xi}_j^2}
        \bigg|_{\bm{\xi}_j=\mathbf{s}_j(\bm{\eta})}\right)\mathrm{null}(\mathbf{A})\\
        &+ \gamma\sum_{j \in \mathcal{J}}\pi_j(\bm{\eta})
        \left(\frac{\partial (f_j(\bm{\eta}) - f(\bm{\eta}))}{\partial\bm{\eta}}\right)^\top
        \frac{\partial (f_j(\bm{\eta}) - f(\bm{\eta}))}{\partial\bm{\eta}}.
    \end{aligned}$}
\end{equation}
From this expression, it becomes clear that the Hessian matrix 
$\frac{\partial^2}{\partial\bm{\eta}^2} f(\bm{\eta})$ is always positive definite when each 
$g_j$ is strictly convex, as the first term is always positive definite under strict 

\newpage 

\noindent convexity of each $g_j$, while the second term is always positive semi-definite, which 
confirms the initial claim.

\begin{remark}[Implementation Details] \label{Rm:Implementation}
    In general, the implicit function $\bm{\ell}$ in the proof of Theorem \ref{Th:TheoremCBGCCG}, 
    and hence the CBF $h$ in \eqref{Eq:CBFCCG}, cannot be expressed in closed form. Under the 
    conditions of Theorem \ref{Th:TheoremCBGCCG}, however, $h$ admits the first-order approximation
    \begin{equation}
        h_\mathrm{l}(\mathbf{p}, t) = h(\mathbf{p}_0, t_0) 
        + \frac{\partial h(\mathbf{p}, t)}{\partial(\mathbf{p}, t)}
        \bigg|_{(\mathbf{p}_0, t_0)}(\mathbf{p}-\mathbf{p}_0, t-t_0),
    \end{equation}
    with $h$ and its gradient at $(\mathbf{p}_0, t_0)$ obtained through 
    \eqref{Eq:Numerical1}-\eqref{Eq:Numerical3} using 
    $\bm{\ell}(\mathbf{p}_0, t_0) = (\bm{\eta}_0, \bm{\lambda}_0)$, the numerical solution of the 
    KKT system at $(\mathbf{p}_0, t_0)$. This solution can be obtained, e.g., as
    \begin{equation}
        \scalebox{0.93}{$
        \begin{aligned}
            \bm{\alpha}_0 &= \underset{\bm{\alpha} \in \mathbb{R}^{\eta-p}}{\arg\min}\,\,
            f\!\left(\tilde{\mathbf{G}}(t_0)^\dagger
            \left(\mathbf{p}_0-\tilde{\mathbf{c}}(t_0)\right) 
            + \mathrm{null}\big(\tilde{\mathbf{G}}(t_0)\big)\bm{\alpha}\right),\\
            \bm{\eta}_0 &= \tilde{\mathbf{G}}(t_0)^\dagger
            \left(\mathbf{p}_0-\tilde{\mathbf{c}}(t_0)\right) 
            + \mathrm{null}\big(\tilde{\mathbf{G}}(t_0)\big)\bm{\alpha}_0,\\
            \bm{\lambda}_0 &= \left(\tilde{\mathbf{G}}(t_0)^\dagger\right)^\top
            \left(\dfrac{\partial f(\bm{\eta})}{\partial\bm{\eta}}
            \bigg|_{\bm{\eta}=\bm{\eta}_0}\right)^\top,
        \end{aligned}$}
        \label{Eq:Implementation}
    \end{equation}
    where computing $\bm{\alpha}_0$ involves an unconstrained convex optimization problem, 
    solvable efficiently with standard solvers.
    
    This first-order approximation, however, may not formally preserve safety. More precisely, 
    consider the time-varying half-space 
    $\mathcal{H}: [t_0, t_\mathrm{f}) \rightrightarrows \mathbb{R}^p$ associated with 
    $h_\mathrm{l}$, defined as
    \begin{equation}
        \mathcal{H}(t) = \{\mathbf{p} \in \mathbb{R}^p : h_\mathrm{l}(\mathbf{p}, t) \geq 0\}.
    \end{equation}
    As $h$ is convex in $\mathbf{p}$, $\mathcal{H}(t_0) \subseteq \mathcal{C}(t_0)$, so 
    no loss of safety occurs at $t_0$. This guarantee, however, does not generally extend to 
    $\mathcal{H}(t) \subseteq \mathcal{C}(t)$ for all $t \in [t_0, t_\mathrm{f})$, as $h$ may 
    not be convex in $t$. 
    
    To robustify the approximation in $t$, consider the estimated obstacle-velocity set over 
    $[t_0, t_\mathrm{f})$, which follows from \eqref{Eq:ObstacleDynamics} as
    \begin{equation}
        \hat{\mathcal{V}}(t) = \mathbf{E}\mathbf{F}\left(\mathbf{\Phi}(t - t_0)\hat{\mathcal{X}}_0
        \oplus \mathbf{\Gamma}(t - t_0)\tilde{\mathcal{W}}\right) \oplus \mathbf{E}\mathcal{W},
    \end{equation}
    which is a CCG $\hat{\mathcal{V}}(t) = \left(\mathbf{G}_{\hat{v}}(t), \mathbf{c}_{\hat{v}}(t), 
    \mathbf{A}_{\hat{v}}, \mathbf{b}_{\hat{v}}, \mathfrak{G}_{\hat{v}}\right)$, with
    \begin{equation}
        \scalebox{0.97}{$
        \begin{aligned}
            \mathbf{G}_{\hat{v}}(t) &= 
            \begin{bmatrix}
                \mathbf{E}\mathbf{F}\mathbf{\Phi}(t - t_0)\mathbf{G}_{\hat{x}, 0} 
                & \mathbf{E}\mathbf{F}\mathbf{\Gamma}(t - t_0)\mathbf{G}_{\tilde{w}}
                & \mathbf{E}\mathbf{G}_{w}
            \end{bmatrix},\\
            \mathbf{c}_{\hat{v}}(t) &= 
            \mathbf{E}\mathbf{F}\mathbf{\Phi}(t - t_0)\mathbf{c}_{\hat{x}, 0} 
            + \mathbf{E}\mathbf{F}\mathbf{\Gamma}(t - t_0)\mathbf{c}_{\tilde{w}} 
            + \mathbf{E}\mathbf{c}_{w},\\
            \mathfrak{G}_{\hat{v}} &= 
            \mathfrak{G}_{\hat{x}, 0}\times\mathfrak{G}_{\tilde{w}}\times\mathfrak{G}_{w}.
        \end{aligned}$}
    \end{equation}
    When the auxiliary estimator used in the finite-horizon scheme yields CCGs with a fixed 
    generator set, as is typically the case, $\mathfrak{G}_{\hat{v}}$ is also fixed, and we can 
    obtain $\bar{\xi} = \max_{\bm{\xi} \in \mathfrak{G}_{\hat{v}}}\|\bm{\xi}\|$ offline, yielding 
    an ellipsoidal approximation $\hat{\mathcal{E}}(t) \supseteq \hat{\mathcal{V}}(t)$ given by
    \begin{equation}
        \scalebox{0.91}{$
        \begin{aligned}
            \hat{\mathcal{E}}(t) &= \left\{\mathbf{G}_{\hat{v}}(t)\bm{\xi} + 
            \mathbf{c}_{\hat{v}}(t): \mathbf{A}_{\hat{v}}\bm{\xi} = \mathbf{b}_{\hat{v}}, 
            \|\bm{\xi}\|^2 \leq \bar{\xi}^2\right\}\\
            &= \left\{\mathbf{G}_{\hat{e}}(t)\bm{\zeta} + \mathbf{c}_{\hat{e}}(t): 
            \left\|\bm{\zeta}\right\|^2 \leq \bar{\xi}^2 
            - \big\|\mathbf{A}_{\hat{v}}^\dagger\mathbf{b}_{\hat{v}}\big\|^2 
            =: \bar{\zeta}^2\right\},
        \end{aligned}$}
    \end{equation}
    with center $\mathbf{c}_{\hat{e}}(t) = \mathbf{G}_{\hat{v}}(t)\mathbf{A}_{\hat{v}}^\dagger
    \mathbf{b}_{\hat{v}} + \mathbf{c}_{\hat{v}}(t)$ and shape governed by 
    $\mathbf{G}_{\hat{e}}(t) = \mathbf{G}_{\hat{v}}(t)\mathrm{null}(\mathbf{A}_{\hat{v}})$. Based 
    on this, we can bound the velocity $v_{\hat{\mathbf{g}}_0}(t)$ of the obstacle along the 
    linearization direction $\hat{\mathbf{g}}_0 = \mathbf{g}_0/\|\mathbf{g}_0\|$, where 
    $\mathbf{g}_0^\top = \frac{\partial}{\partial\mathbf{p}}h(\mathbf{p}, t)
    |_{(\mathbf{p}_0, t_0)}$, as follows:
    \begin{equation}
        \scalebox{0.95}{$
        \begin{aligned}
            v_{\hat{\mathbf{g}}_0}(t) &\leq \max_{\|\bm{\zeta}\| \leq \bar{\zeta}}
            \hat{\mathbf{g}}_0^\top(\mathbf{c}_{\hat{e}}(t) + \mathbf{G}_{\hat{e}}(t)\bm{\zeta})\\
            &= \hat{\mathbf{g}}_0^\top\mathbf{c}_{\hat{e}}(t) + \bar{\zeta}
            \left\|\mathrm{null}(\mathbf{A}_{\hat{v}})^\top\mathbf{G}_{\hat{v}}(t)^\top 
            \hat{\mathbf{g}}_0\right\|\\
            &\leq \hat{\mathbf{g}}_0^\top\mathbf{c}_{\hat{e}}(t) 
            + \bar{\zeta}\|\mathbf{G}_{\hat{v}}(t)\|\\
            &\leq \hat{\mathbf{g}}_0^\top\mathbf{c}_{\hat{e}}(t) + \bar{\zeta}
            \left((\kappa_1 + \kappa_2)\exp(\|\mathbf{F}\|(t-t_0)) + (\kappa_3-\kappa_2)\right)\\
            &=: \bar{v}_{\hat{\mathbf{g}}_0}(t),
        \end{aligned}$}
    \end{equation}
    where \hfill the \hfill second \hfill inequality \hfill follows \hfill from \hfill the \hfill 
    orthonormality

    \newpage

    \noindent of $\mathrm{null}(\mathbf{A}_{\hat{v}})$, and the final inequality follows from the 
    matrix-exponential growth bound $\|\mathbf{\Phi}(s)\| \leq \exp(\|\mathbf{F}\|s)$ together with 
    the triangle inequality applied to the block structure of $\mathbf{G}_{\hat{v}}(t)$, with 
    constants $\kappa_1 = \|\mathbf{E}\mathbf{F}\|\|\mathbf{G}_{\hat{x},0}\|$, 
    $\kappa_2 = \|\mathbf{E}\mathbf{F}\|\|\mathbf{G}_{\tilde{w}}\|/\|\mathbf{F}\|$, and 
    $\kappa_3 = \|\mathbf{E}\mathbf{G}_w\|$. Integrating this velocity bound over time yields the 
    robustified CBF $h_\mathrm{r}: \mathbb{R}^p\times[t_0, t_\mathrm{f}) \to \mathbb{R}$ defined as
    \begin{equation}
        h_\mathrm{r}(\mathbf{p}, t) = h(\mathbf{p}_0, t_0) 
        + \mathbf{g}_0^\top(\mathbf{p}-\mathbf{p}_0) 
        - \|\mathbf{g}_0\|\int_{t_0}^{t}\bar{v}_{\hat{\mathbf{g}}_0}(\tau)d\tau,
    \end{equation}
    where the above integral admits a simple closed-form expression. The function $h_\mathrm{r}$ 
    defines a time-varying safe half-space 
    $\mathcal{H}_\mathrm{r}: [t_0, t_\mathrm{f}) \rightrightarrows \mathbb{R}^p$ satisfying 
    $\mathcal{H}_\mathrm{r}(t) \subseteq \mathcal{C}(t)$ for all $t \in [t_0, t_\mathrm{f})$, hence 
    formally preserving safety.
    
    Also, we note that this paper considers a continuous implementation of the safety 
    filters over each sampling interval. For a sample-and-hold implementation, however, the 
    additional discrepancy would need to be further robustified, as in \cite{breeden2021control}.
\end{remark}

\begin{remark}[Complexity]
    Following Remark \ref{Rm:Implementation} and recalling Algorithm \ref{Alg:Solution}, it is 
    clear that the complexity of the proposed framework is dominated, at each sampling instant, by 
    solving the unconstrained convex program in \eqref{Eq:Implementation} once per obstacle, since 
    the remaining steps in Algorithm \ref{Alg:Solution} require only closed-form operations. The 
    problem in \eqref{Eq:Implementation} has dimension $\eta - p$, where $\eta$ grows linearly with 
    the horizon length $N$ of the finite-horizon scheme. Since its cost function is continuously 
    differentiable and convex, a standard gradient method converges at rate $O(1/n)$, or $O(1/n^2)$ 
    with an accelerated scheme, where $n$ is the iteration count \cite{nesterov2013gradient}. 
    Solved independently per obstacle, this stage scales linearly in the number of obstacles $M$.
\end{remark}

\begin{remark}[Ensuring Global Safety] \label{Rm:GlobalSafety}
    As discussed in Section \ref{Sec:Solution1}, global safety is guaranteed if, at every sampling 
    time $t_k$, the condition 
    $\mathbf{p}_k \in \mathrm{int}(\cap_{i \in \mathcal{I}}\mathcal{C}_{i, k}(t_k))$ holds. If we 
    employ the infinite-horizon estimation scheme described at the beginning of Section 
    \ref{Sec:Solution2}, then this condition is preserved at all sampling instants if it 
    holds initially at $t = 0$. This follows from two facts: (i) if at the time $t_k$ we have 
    $\mathbf{p}_k \in \mathrm{int}(\cap_{i \in \mathcal{I}}\mathcal{C}_{i, k}(t_k))$, then the 
    trajectory of the agent satisfies 
    $\bm{\varphi}(t) \in \mathrm{int}(\cap_{i \in \mathcal{I}}\mathcal{C}_{i, k}(t))$ for all 
    $t \in [t_k, t_{k+1})$ (ensured by the strict inequality in \eqref{Eq:SafetyDetail}), and (ii) 
    the infinite-horizon estimator, along with the CCG-to-CBF conversion process, ensures that 
    $\mathcal{C}_{i, k}(t_{k+1}^-) \subseteq \mathcal{C}_{i, k+1}(t_{k+1})$ for all 
    $i \in \mathcal{I}$. However, using the finite-horizon estimator to keep a fixed computational 
    load means that $\mathbf{p}_{k+1} \in \mathrm{int}(\mathcal{C}_{i, k}(t_{k+1}^-))$ does not 
    necessarily imply that $\mathbf{p}_{k+1} \in \mathrm{int}(\mathcal{C}_{i, k+1}(t_{k+1}))$, due 
    to the conservatism introduced by discarding the measurement from time $t_{k-N}$. Several 
    alternatives are possible to still use the finite-horizon estimator while maintaining 
    theoretical guarantees of safety, such as rejecting the new estimate produced by this estimator
    whenever the condition is not satisfied and continuing the propagation of the estimate from 
    $t_{k+1}^-$ (before the update) over the next interval. Nevertheless, we note that the 
    finite-horizon estimator in \cite{rego2024explicit}, which uses a Luenberger observer as the 
    conservative estimator, can be designed using a 
    \textit{deadbeat} gain to remove the conservatism in $n_i$ time steps (under an observability 
    assumption), as shown in \cite{silvestre2017set} for fault detection. Overall, this discussion 
    highlights the importance for obstacle estimates to never contain the agent position to 
    preserve theoretical guarantees of global safety.
\end{remark}


\section{Simulation Results} \label{Sec:Results}

This section presents three simulation examples that demonstrate the framework proposed in this 
paper. The first example considers an agent navigating in a known static environment and serves to 
illustrate different problem geometries and the CCG-to-CBF conversion procedure. The remaining two 
examples involve dynamic obstacles with uncertain linear dynamics and demonstrate the overall 
framework proposed in the paper.

\begin{example} \label{Ex:Example1}
    Fig. \ref{Fig:Example1} exemplifies the safe navigation of a rigid-body agent in a known 
    static environment with two obstacles, where both the agent and obstacle body sets are CCGs. 
    Two different agent geometries are considered: an ellipsoidal agent, shown in Fig. 
    \ref{Fig:Example1} (a), and a polytopic agent, illustrated in Fig. \ref{Fig:Example1} (b). In 
    this example, the agent has single-integrator dynamics:
    \begin{equation}
        \Dot{\mathbf{p}} = \mathbf{z},
    \end{equation}
    with $\mathbf{p}, \mathbf{z} \in \mathbb{R}^2$, and the agent is controlled using the approach 
    proposed in this paper, without the estimation component, to safely guide it to a goal point 
    $\bar{\mathbf{p}} \in \mathbb{R}^2$. Convergence to the goal is captured by a nominal 
    controller $\mathbf{k}_\mathrm{d}: \mathbb{R}^2 \to \mathbb{R}^2$, defined as
    \begin{equation}
        \mathbf{k}_\mathrm{d}(\mathbf{p}) = K(\mathbf{p}-\bar{\mathbf{p}})
    \end{equation}
    for all $\mathbf{p} \in \mathbb{R}^2$, where $K$ is a negative gain. 

    The obstacles are modeled as CCGs with mixed geometry: one of them is defined as the 
    intersection of a polytope with an ellipsoid, and the other is formed by the intersection of 
    two ellipsoids. An ellipsoid $\mathcal{E} \subset \mathbb{R}^p$ is a CCG of the form
    \begin{equation}
        \mathcal{E} = \left(\mathbf{G}_e, \mathbf{c}_e, [\,\,], [\,\,], \mathfrak{G}_e\right),
    \end{equation}
    where $\mathfrak{G}_e$ is the unit $\ell_2$-ball in $\mathbb{R}^p$. The generator set 
    $\mathfrak{G}_e$ can be described as the zero-sublevel set of a function 
    $g_e: \mathbb{R}^p \to \mathbb{R}$, defined for all $\bm{\xi}_e \in \mathbb{R}^p$ as
    \begin{equation}
        g_e(\bm{\xi}_e) = \frac{1}{2}\|\bm{\xi}_e\|^2 - \frac{1}{2},
    \end{equation}
    which is twice continuously differentiable and strictly convex. A polytope (or CZ) 
    $\mathcal{Z} \subset \mathbb{R}^p$ is a CCG of the form
    \begin{equation}
        \mathcal{Z} = \left(\mathbf{G}_z, \mathbf{c}_z, \mathbf{A}_z, \mathbf{b}_z, 
        \mathfrak{G}_z\right),
    \end{equation}
    where the set $\mathfrak{G}_z$ is the unit $\ell_\infty$-ball in $\mathbb{R}^{\xi_z}$. 
    However, since the \hfill $\ell_\infty$-norm \hfill is \hfill not \hfill differentiable \hfill 
    neither \hfill strictly \hfill convex, \hfill we cannot directly derive a single suitable 
    generator function. To address this, we can decompose $\mathfrak{G}_z$ as 
    $\mathfrak{G}_z = \mathcal{G}_z^{\xi_z} = \times_{j=1}^{\xi_z}\mathcal{G}_z$, where 
    $\mathcal{G}_z$ is the unit $\ell_2$-ball in $\mathbb{R}$, which can be described as the 
    zero-sublevel set of a function $g_z : \mathbb{R} \to \mathbb{R}$ defined as
    \begin{equation}
        g_z(s) = \frac{1}{2}s^2 - \frac{1}{2},
    \end{equation}
    which is twice continuously differentiable and strictly convex.

    The simulations are conducted with a sampling period of $T_s = \SI{0.1}{s}$. Over each 
    sampling interval, the safe controller defined in \eqref{Eq:SolutionSafetyFilter1} is applied 
    to the agent, with design parameters $\varepsilon_\beta = \varepsilon_\gamma = 0.1$ and 
    $\alpha(s) = \bar{\alpha}s$ for all $s \in \mathbb{R}$, with $\bar{\alpha} = 10$. The 
    obstacle-specific CBFs are approximated as in Remark \ref{Rm:Implementation}, by using a 
    first-order Taylor expansion around the position of the agent at each sampling instant. The 
    unconstrained optimization step in \eqref{Eq:Implementation} is solved using a built-in 
    quasi-Newton method in \textsc{Matlab}, with an average computation time of \SI{3}{ms} per 
    obstacle per sampling step. As shown in Fig. \ref{Fig:Example1}, collision-free 
    
    \newpage

\begin{figure}[H]
    \centering 
    \subfloat[Agent with ellipsoidal geometry]{
        \includegraphics[width=0.98\linewidth]{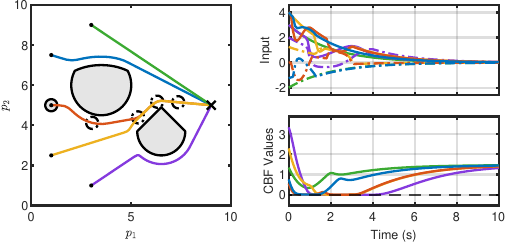}}\\
    \subfloat[Agent with polytopic geometry]{
        \includegraphics[width=0.98\linewidth]{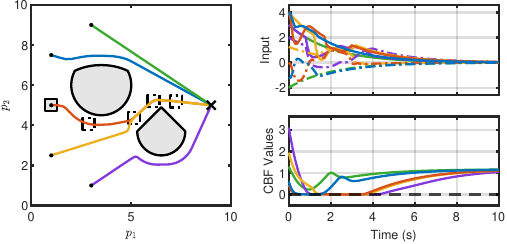}}
    \caption{Navigation of an agent with single-integrator dynamics and ellipsoidal and polytopic 
    geometry around obstacles with mixed ellipsoidal and polytopic geometry. The left-hand plots 
    show trajectories from different initial positions, and the right-hand plots illustrate the 
    respective time evolution of the control inputs and overall CBF values. The colors in the 
    right-hand plots correspond to those of the associated trajectories on the left-hand plots. In 
    the input plots, solid lines represent the first input component and dotted lines the second.}
    \label{Fig:Example1}
\end{figure}
    
    \noindent motion is achieved for all the initial agent positions, confirmed by the 
    nonnegativity of the overall CBF values over time. The conservativeness of the resulting 
    trajectories depends on the desired tolerances $\varepsilon_\beta$ and $\varepsilon_\gamma$, as 
    well as the decay rate $\bar{\alpha}$.
\end{example}

\begin{example} \label{Ex:Example2}
    Fig. \ref{Fig:Example2} exemplifies the navigation of an ellipsoidal agent in an environment 
    cluttered with dynamic obstacles with uncertain linear dynamics, showcasing the overall 
    framework proposed in this paper. In this example, the agent also has single-integrator 
    dynamics and is guided toward a target point using the same nominal controller from Example 
    \ref{Ex:Example1}.

    The environment contains two obstacles with different geometries (one polytopic and one 
    ellipsoidal), and each obstacle is known to evolve with an uncertain velocity. Specifically, 
    the motion of obstacle $i$ is governed by
    \begin{equation}
        \Dot{\mathbf{o}}_i = \mathbf{w}_i,
    \end{equation}
    with $\mathbf{o}_i, \mathbf{w}_i \in \mathbb{R}^2$, where $\mathbf{w}_i$ is an uncertain 
    velocity, known to belong to an ellipsoidal set
    $\mathcal{W} = (\mathbf{G}_w, \mathbf{0}, [\,\,], [\,\,], \mathfrak{G}_w) \subset 
    \mathbb{R}^2$. In the simulation, the actual velocities are constant and given by 
    $\mathbf{w}_1 = -\mathbf{w}_2 = [0\,\,-0.5]^\top$. The measurement model is
    \begin{equation}
        \mathbf{y}_{i, k} = \mathbf{o}_{i, k} + \mathbf{v}_{i, k},
    \end{equation}
    where the measurement noise $\mathbf{v}_{i, k}$ is a random variable drawn uniformly from an 
    ellipsoid $\mathcal{V} = (\mathbf{G}_v, \mathbf{0}, [\,\,], [\,\,], \mathfrak{G}_v) \subset 
    \mathbb{R}^2$.

    The simulation setup matches that of Example \ref{Ex:Example1}, with the additional parameters
    $\mathbf{G}_w = 0.5\mathbf{I}$ and $\mathbf{G}_v = 0.2\mathbf{I}$. Moreover, the finite 
    horizon estimation scheme described in Section \ref{Sec:Solution2} is implemented with a 
    horizon of $N = 5$, with the conservative 

    \newpage
    
    \noindent estimator being defined as 
    $\bar{\mathcal{X}}_{i, k} = \mathbf{y}_{i, k} - \mathcal{V}$. For the considered horizon 
    length, the unconstrained optimization step in \eqref{Eq:Implementation} takes an average of 
    \SI{5}{ms} per obstacle per sampling step.
    
    As illustrated in Fig. \ref{Fig:Example2} (a), collision-free motion is achieved throughout 
    the simulation, confirmed by the nonnegative overall CBF values over time. Fig. 
    \ref{Fig:Example2} (b) shows snapshots of the agent and the estimated obstacle sets at 
    selected time instants. The orange dotted contours denote the estimated CCG obstacle  sets, 
    which enclose the actual obstacle sets. During the first $N$ sampling steps, the state 
    estimates are obtained directly from the conservative estimator, resulting in looser estimates. 
    After the initial transient, the estimator produces steadier and tighter estimates, with the 
    degree of tightness being determined by the sets $\mathcal{W}$ and $\mathcal{V}$, as well as 
    the horizon length $N$.
\end{example}

\begin{example} \label{Ex:Example3}
    Finally, Fig. \ref{Fig:Example3} illustrates a scenario similar to that from Example 
    \ref{Ex:Example2}, but where now both the agent and obstacles have second-order dynamics. 
    Specifically, the agent dynamics are inspired by a satellite subject to a gravitational force 
    and are described by the second-order strict-feedback system
    \begin{equation}
        \begin{aligned}
            \Dot{\mathbf{p}} &= \mathbf{z},\\
            \Dot{\mathbf{z}} &= \mathbf{f}_1(\mathbf{p}) + \mathbf{u},
        \end{aligned}
    \end{equation}
    with $\mathbf{p}, \mathbf{z}, \mathbf{u} \in \mathbb{R}^2$, where 
    $\mathbf{f}_1(\mathbf{p}) = (\|\mathbf{p}\|+0.1)^{-3}\mathbf{p}$ represents a gravity-like 
    effect. Similar to the previous example, the agent is controlled using the overall approach 
    proposed in this paper, with the addition of the backstepping-based design in Section 
    \ref{Sec:Solution1A}, to safely guide it to a target point $\bar{\mathbf{p}}$. Convergence to 
    the goal is captured by a nominal controller 
    $\mathbf{k}_{\mathrm{d}, 1}: \mathbb{R}^2\times\mathbb{R}^2 \to \mathbb{R}^2$, defined for all 
    $(\mathbf{p}, \mathbf{z}) \in \mathbb{R}^2\times\mathbb{R}^2$ as
    \begin{equation}
        \mathbf{k}_{\mathrm{d}, 1}(\mathbf{p}, \mathbf{z}) = -\mathbf{f}_1(\mathbf{p})
        + K_1(\mathbf{z} - K(\mathbf{p}-\bar{\mathbf{p}})),
    \end{equation}
    where $K$ and $K_1$ are negative gains. 
    
    The environment is populated with two obstacles with the same geometries from Example 
    \ref{Ex:Example2}, and each obstacle evolves subject to a random acceleration. Specifically, 
    the motion of obstacle $i$ is governed by
    \begin{equation}
        \begin{aligned}
            \Dot{\mathbf{o}}_i &= \mathbf{z}_i,\\
            \Dot{\mathbf{z}}_i &= \mathbf{w}_i,
        \end{aligned}
    \end{equation}
    with $\mathbf{o}_i, \mathbf{z}_i, \mathbf{w}_i \in \mathbb{R}^2$, where $\mathbf{w}_i$ is a 
    random variable drawn uniformly from an ellipsoid 
    $\mathcal{W} = (\mathbf{G}_w, \mathbf{0}, [\,\,], [\,\,], \mathfrak{G}_w) \subset 
    \mathbb{R}^2$. The measurement model is of the form
    \begin{equation}
        \mathbf{y}_{i, k} = (\mathbf{o}_{i, k}, \mathbf{z}_{i, k})  + \mathbf{v}_{i, k},
    \end{equation}
    where the measurement noise $\mathbf{v}_{i, k}$ is a random variable drawn uniformly from an 
    ellipsoid $\mathcal{V} = (\mathbf{G}_v, \mathbf{0}, [\,\,], [\,\,], \mathfrak{G}_v) \subset 
    \mathbb{R}^4$.

    The simulation setup matches that of the previous examples, with 
    $\mathbf{G}_w = 0.5\mathbf{I}_2$, $\mathbf{G}_v = 0.2\mathbf{I}_4$, the backstepping 
    parameters $\varsigma = 0.1$, $\bar{\sigma} = 10$, and the CBF rate 
    $\alpha_1(s) = \bar{\alpha}_1s$ for all $s \in \mathbb{R}$, with $\bar{\alpha}_1 = 10$. The 
    finite-horizon estimator is implemented with a horizon length $N = 5$, with the conservative 
    estimator again defined as $\bar{\mathcal{X}}_{i, k} = \mathbf{y}_{i, k} - \mathcal{V}$. For 
    the considered horizon length, the unconstrained optimization step in 
    \eqref{Eq:Implementation} takes an average of \SI{8}{ms} per obstacle per sampling step. 
    
    As illustrated in Fig. \ref{Fig:Example3} (a), collision-free motion is achieved throughout 
    the simulation, confirmed by the nonnegative top-level CBF values over time. Additionally, Fig. 
    \ref{Fig:Example3} (b) presents snapshots of the agent and the estimated CCG obstacle sets 
    
    \newpage

\begin{figure}[H]
    \centering 
    \subfloat[Trajectories and temporal profiles]{
        \includegraphics[width=0.98\linewidth]{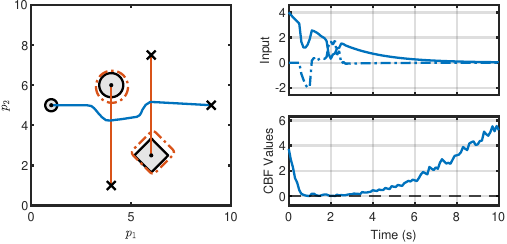}}\\
    \subfloat[Snapshots]{
        \includegraphics[width=0.98\linewidth]{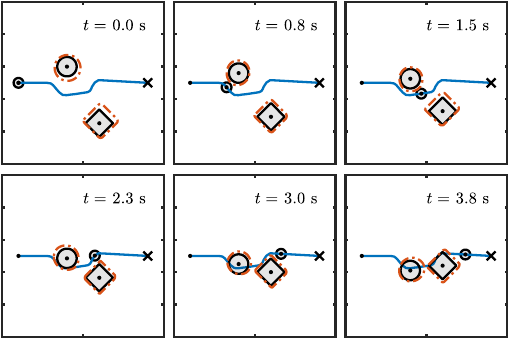}}
    \caption{Safe navigation of an ellipsoidal agent with single-integrator dynamics around 
    moving obstacles with uncertain linear dynamics. The left-hand plot in subfigure (a) 
    illustrates the agent and obstacle trajectories, and the right-hand plots shows the 
    respective time evolution of the control input and overall CBF values. In the input plot, 
    the solid line represents the first input component and the dotted line the second. 
    Subfigure (b) presents snapshots of the agent and obstacle configurations at selected time 
    instants. The orange dotted contours illustrate the estimated CCG obstacle sets used for 
    collision avoidance.}
    \label{Fig:Example2}
\end{figure}
    
    \noindent at selected time instants. As in Example \ref{Ex:Example2}, during the first $N$ 
    sampling steps, the state estimates are obtained directly from the conservative estimator, 
    resulting in looser estimates. After this initial phase, the finite-horizon estimator yields 
    steadier and tighter estimates. Also, extending to an agent with second-order dynamics results 
    in smoother trajectories.
\end{example}


\section{Conclusion} \label{Sec:Conclusion}

This paper introduced a framework for the safe navigation of a controlled agent in environments 
cluttered with obstacles governed by uncertain linear dynamics, combining CBF-based safety 
filtering with CCG-based guaranteed state estimation. The main contribution is a procedure for 
converting CCG-valued obstacle estimates into CBFs via a convex optimization problem, whose 
validity follows from the Implicit Function Theorem. Since CCGs support Minkowski sums, the 
framework naturally handles rigid-body agents, and its effectiveness was demonstrated through 
several simulation examples.

Future directions include extending the framework to nonlinear obstacle dynamics via HZs and 
functional decomposition, formally incorporating input bounds through the optimal-decay formulation 
or backup CBF approaches, and extending the approach to cooperative safe navigation.

\begin{figure}[H]
        \centering 
        \subfloat[Trajectories and temporal profiles]{
            \includegraphics[width=0.98\linewidth]{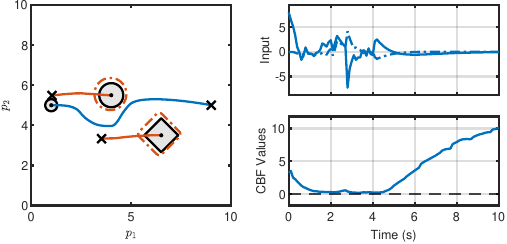}}\\
        \subfloat[Snapshots]{
            \includegraphics[width=0.98\linewidth]{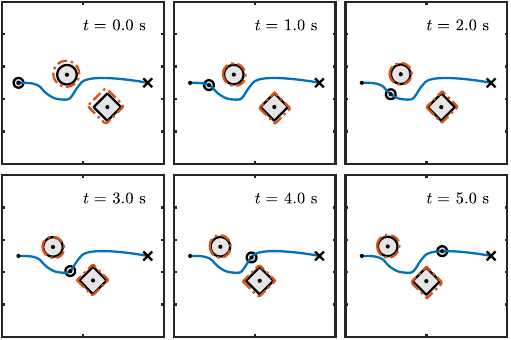}}
        \caption{Navigation of an ellipsoidal agent with second-order strict-feedback dynamics 
        around moving obstacles with uncertain linear dynamics. The left-hand plot in subfigure 
        (a) illustrates the agent and obstacle trajectories, and the right-hand plots shows the 
        respective time evolution of the control input and top-level CBF values. In the input 
        plot, the solid line represents the first input component and the dotted line the second. 
        Subfigure (b) presents snapshots of the agent and obstacle configurations at selected time 
        instants. The orange dotted contours illustrate the estimated CCG obstacle sets used for 
        avoidance.}
        \label{Fig:Example3}
    \end{figure}


\bibliographystyle{ieeetr}
\bibliography{Refs}
\vspace{-4mm}
\begin{IEEEbiography}[
{\includegraphics[width=1in,height=1.25in,keepaspectratio]{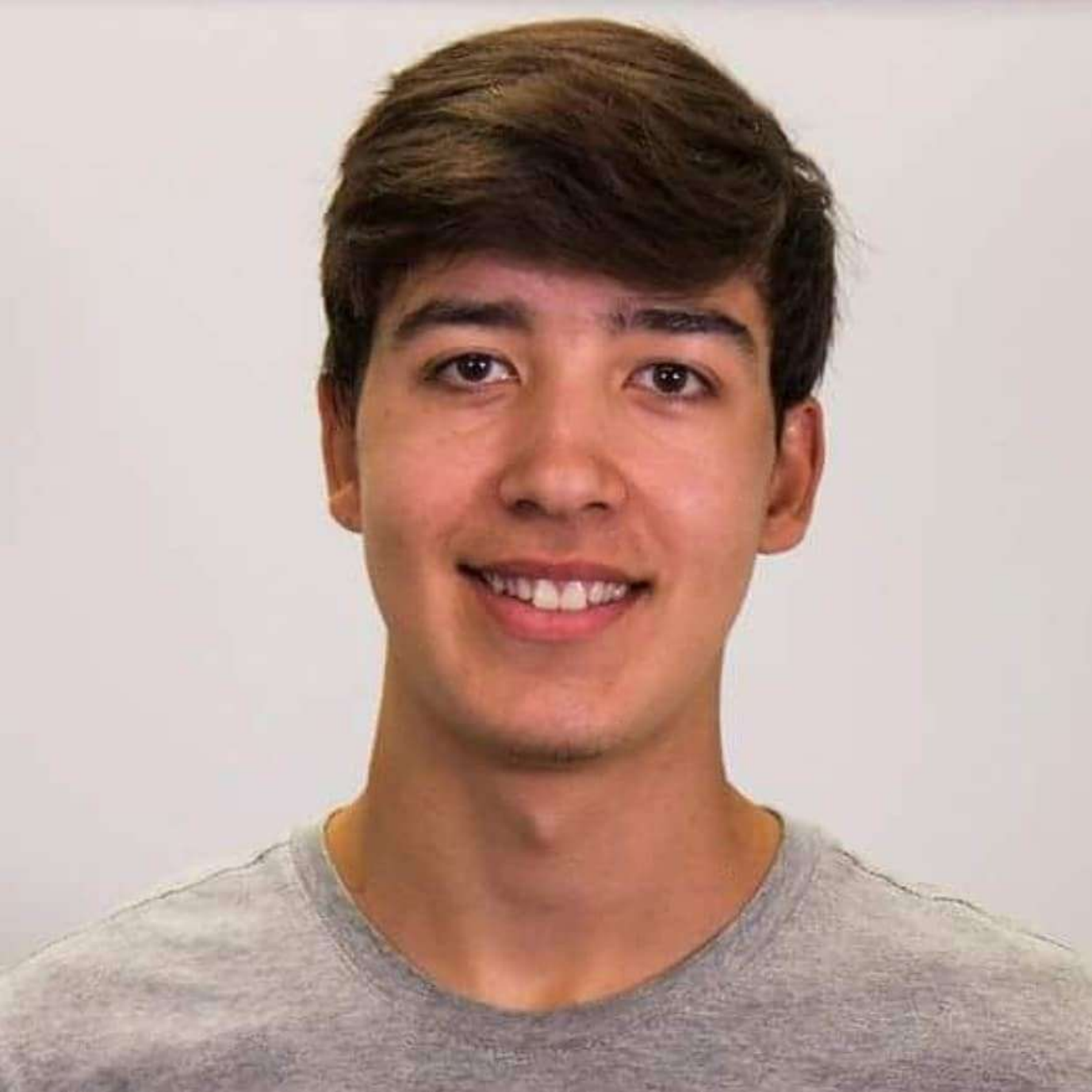}}]
{Hugo Matias} received his B.Sc. and M.Sc. degrees in Electrical and Computer Engineering from the 
Instituto Superior Técnico, Lisbon, Portugal, in 2021 and 2023, respectively. He is a doctoral 
candidate in Electrical and Computer Engineering, specializing in Systems, Decision and Control, at 
the NOVA School of Science and Technology, Caparica, Portugal. His research interests include 
safety-critical control, nonlinear optimization, and guaranteed state estimation.
\end{IEEEbiography}
\vspace{-4mm}
\begin{IEEEbiography}[
{\includegraphics[width=1in,height=1.25in,keepaspectratio]{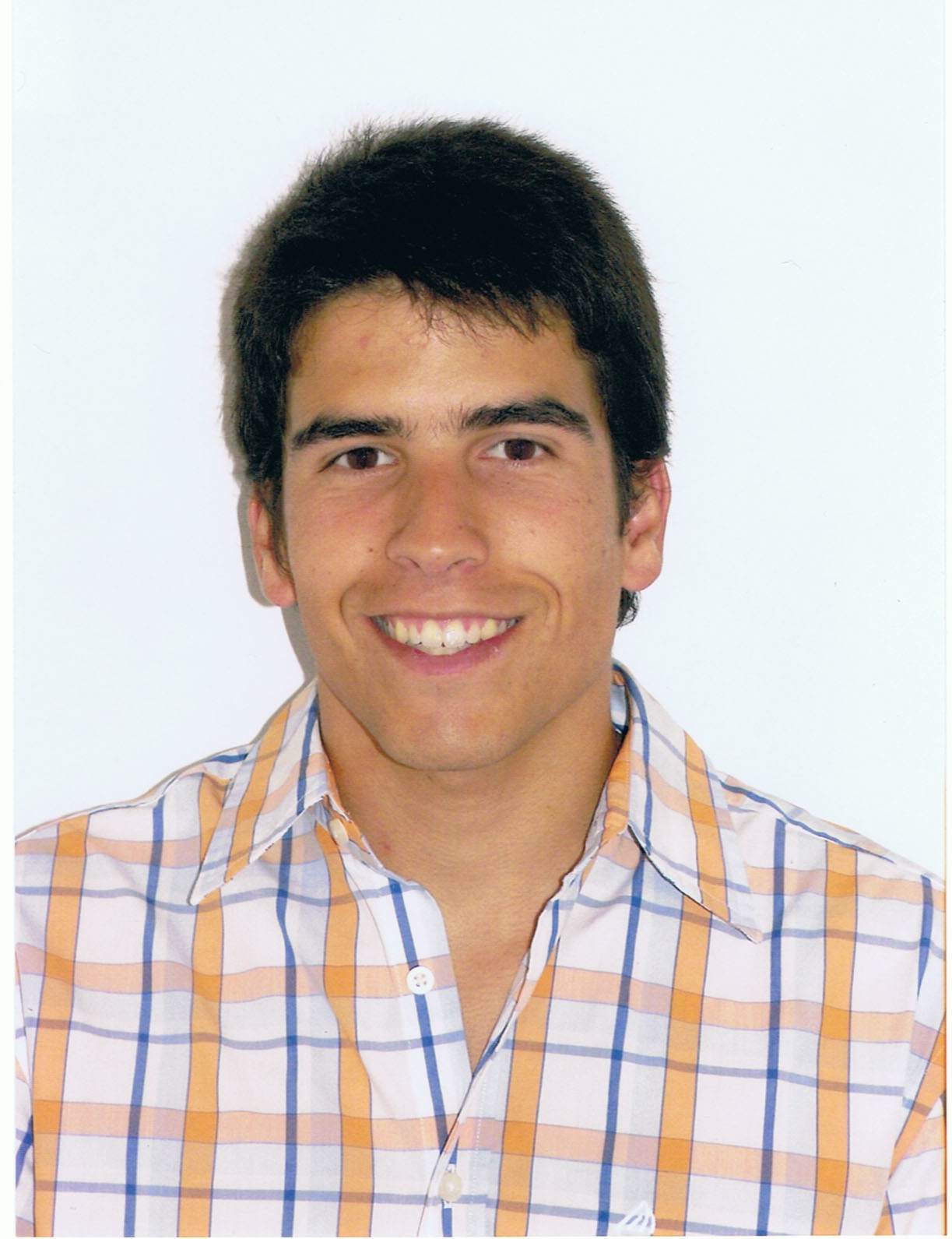}}]
{Daniel Silvestre} received his B.Sc. in Computer Networks in 2008 from Instituto Superior 
Técnico, Lisbon, Portugal, and his M.Sc. in Advanced Computing in 2009 from the Imperial College 
London, London, United Kingdom. In 2017, he received his Ph.D. (with the highest honors) in 
Electrical and Computer Engineering from the former university. Currently, he is with the NOVA 
School of Science and Technology, Caparica, Portugal, and with the Institute for Systems and 
Robotics, Instituto Superior Técnico, Lisbon, Portugal. His research interests include fault 
detection and isolation, distributed systems, guaranteed state estimation, computer networks, 
optimal control and nonlinear optimization.
\end{IEEEbiography}

\end{document}